\documentclass[10pt]{article}
\usepackage{fullpage}

\usepackage{varioref}

\usepackage{amsmath}
\usepackage{amssymb}
\usepackage{amsthm}

\usepackage{units}

\theoremstyle{plain}
\newtheorem{thm}{\protect\theoremname}[section]
  \theoremstyle{definition}
  \newtheorem{defn}[thm]{\protect\definitionname}
  \theoremstyle{remark}
  \newtheorem{rem}[thm]{\protect\remarkname}
  \theoremstyle{plain}
  \newtheorem{lem}[thm]{\protect\lemmaname}
  \theoremstyle{plain}
  \newtheorem{fact}[thm]{\protect\factname}
  \theoremstyle{plain}  
  \newtheorem{example}[thm]{\protect\examplename}
  \theoremstyle{plain}  
  \newtheorem{cor}[thm]{\protect\corname}
  \theoremstyle{plain}    

\makeatother

  \providecommand{\definitionname}{Definition}
  \providecommand{\lemmaname}{Lemma}
  \providecommand{\remarkname}{Remark}
  \providecommand{\theoremname}{Theorem}
  \providecommand{\factname}{Fact}
  \providecommand{\examplename}{Example}  
  \providecommand{\corname}{Corollary}    

\usepackage{xargs}[2008/03/08]
\usepackage{setspace}

\usepackage{graphicx}

\begin{document}

\global\long\def\N{\mathbb{N}}
\global\long\def\C{\mathbb{C}}
\global\long\def\Z#1{\mathbb{Z}/#1\mathbb{Z}}

\global\long\def\A{\mathcal{A}}
\global\long\def\H{\mathcal{H}}
\global\long\def\C{\mathcal{C}}
\global\long\def\T{\mathcal{T}}
\global\long\def\O{\mathcal{O}}
\global\long\def\parity{\oplus}
\global\long\def\hw{\mathrm{hw}}
\global\long\def\hadw{\mathrm{hadw}}
\global\long\def\apex{\mathrm{apex}}
\global\long\def\supp{\mathrm{supp}}
\global\long\def\mod{\mathrm{mod}}

\global\long\def\sigHW#1{\mathtt{HW_{#1}}}
\global\long\def\pass{\mathtt{PASS}}
\global\long\def\act{\mathtt{ACT}}
\global\long\def\pre{\mathtt{PRE}}
\global\long\def\odd{\mathtt{ODD}}
\global\long\def\even{\mathtt{EVEN}}

\global\long\def\phione{\varphi_{\mathit{one}}}
\global\long\def\phiprop{\varphi_{\mathit{prop}}}

\global\long\def\Holant{\mathrm{Holant}}
\global\long\def\val{\mathrm{val}}
\global\long\def\Sig{\mathrm{Sig}}

\global\long\def\pSub{\mathrm{\#Sub}}
\newcommandx\pClique[1][usedefault, addprefix=\global, 1=]{\mathrm{#1Clique}}
\newcommandx\pPart[1][usedefault, addprefix=\global, 1=]{\mathrm{#1PartitionedSub}}
\newcommandx\pGrid[1][usedefault, addprefix=\global, 1=]{\mathrm{#1GridTiling}}
\global\long\def\PerfMatch{\mathrm{PerfMatch}}
\global\long\def\perm{\mathrm{perm}}

\global\long\def\FP{\mathsf{FP}}
\global\long\def\sharpP{\mathsf{\#P}}
\global\long\def\XP{\mathsf{XP}}
\global\long\def\FPT{\mathsf{FPT}}
\newcommandx\Wone[1][usedefault, addprefix=\global, 1=]{\mathsf{#1W[1]}}
\newcommandx\ETH[1][usedefault, addprefix=\global, 1=]{\mathsf{#1ETH}}
\renewcommandx\ETH[1][usedefault, addprefix=\global, 1=]{\mathsf{#1ETH}}
\global\long\def\ModP#1{\mathsf{Mod}_{#1}\mathsf{P}}

\global\long\def\leqFpt{\leq_{\mathit{fpt}}}
\global\long\def\leqFptT{\leq_{\mathit{fpt}}^{T}}
\global\long\def\leqFptPar{\leq_{\mathit{fpt}}^{\mathit{pars}}}
\global\long\def\leqPLin{\leq_{p}^{\mathit{lin}}}
\global\long\def\leqFptLin{\leq_{\mathit{fpt}}^{\mathit{lin}}}

\newcommand*\texEmpty{\vcenter{\hbox{\includegraphics[width=0.375cm]{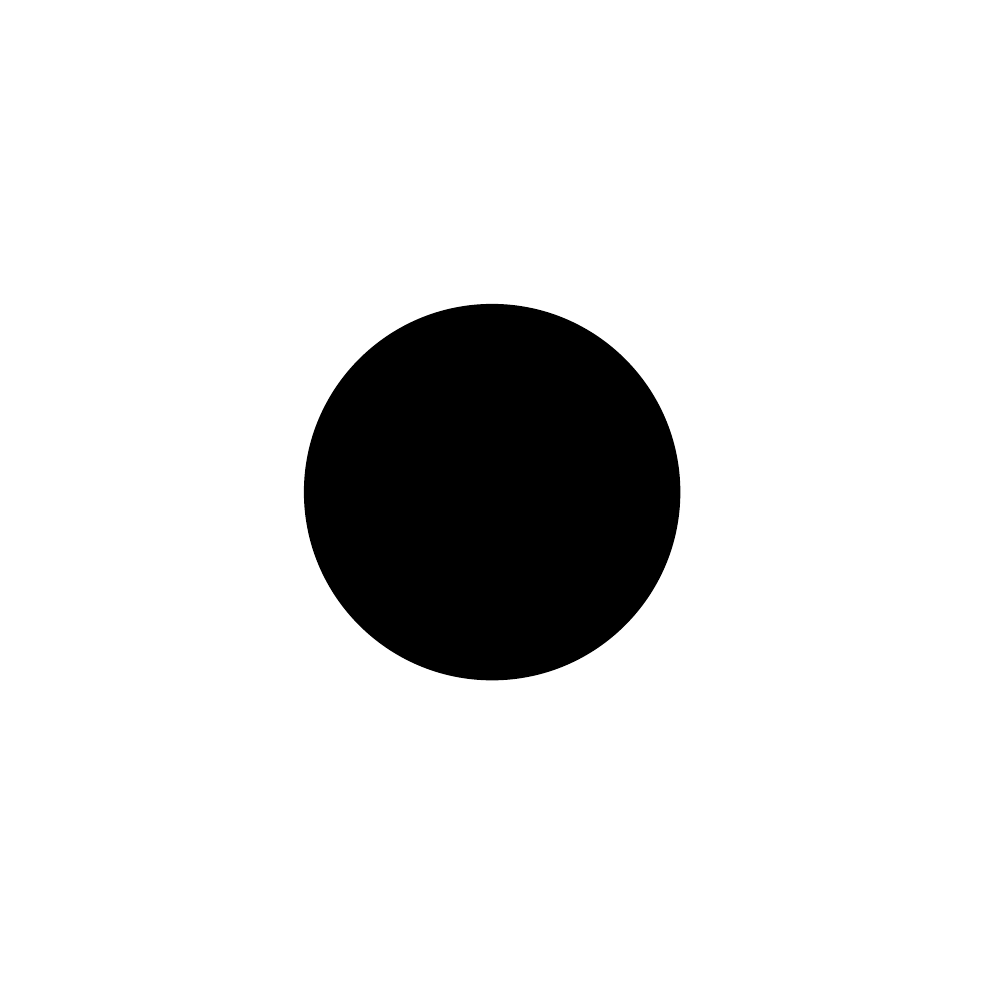}}}}\global\long\def\symEmpty{\texEmpty}
\newcommand*\texNS{\vcenter{\hbox{\includegraphics[width=0.375cm]{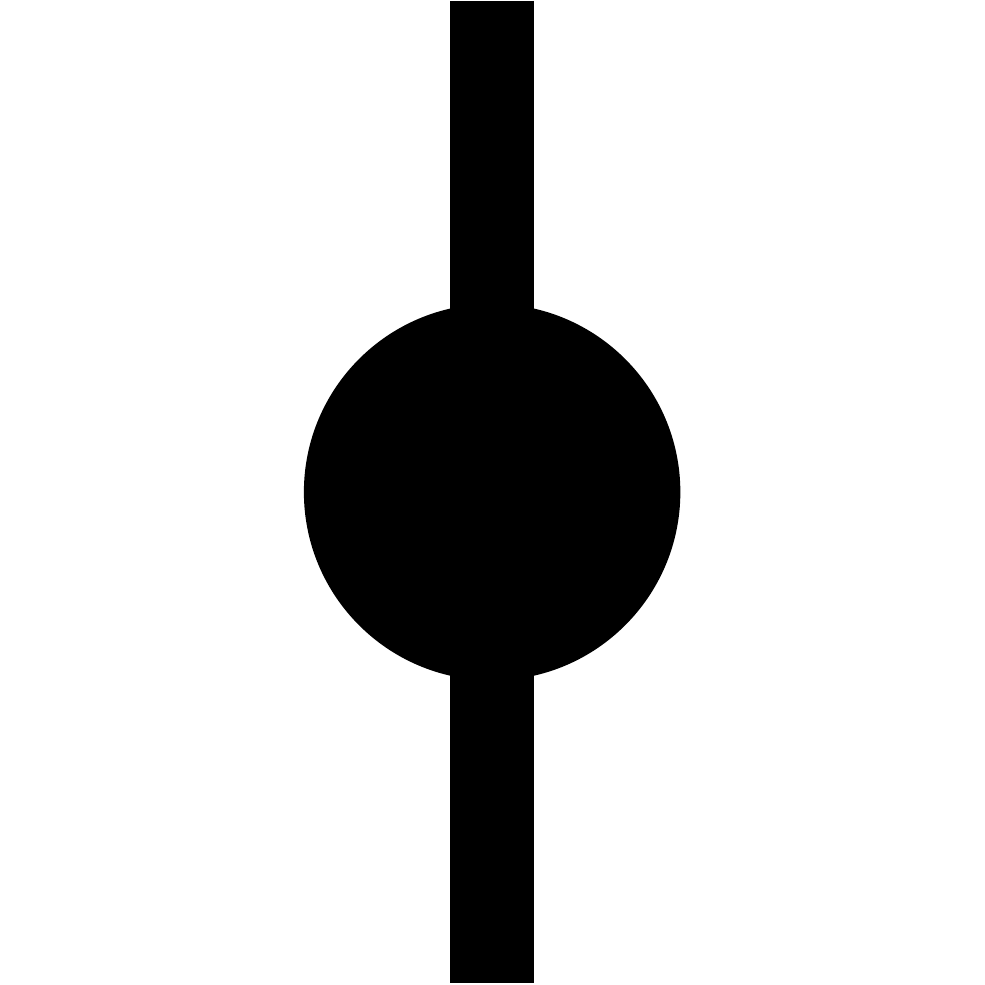}}}}\global\long\def\symNS{\texNS}
\newcommand*\texWE{\vcenter{\hbox{\includegraphics[width=0.375cm]{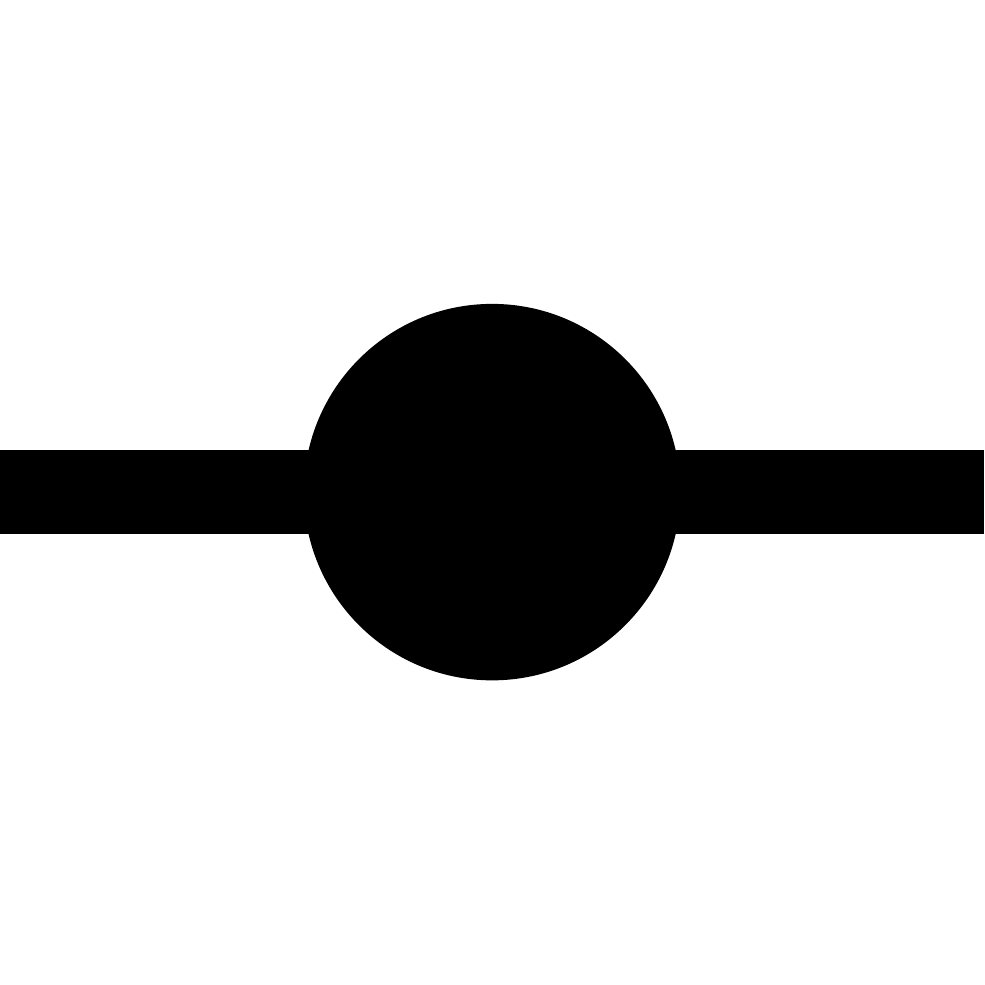}}}}\global\long\def\symWE{\texWE}
\newcommand*\texNSWE{\vcenter{\hbox{\includegraphics[width=0.375cm]{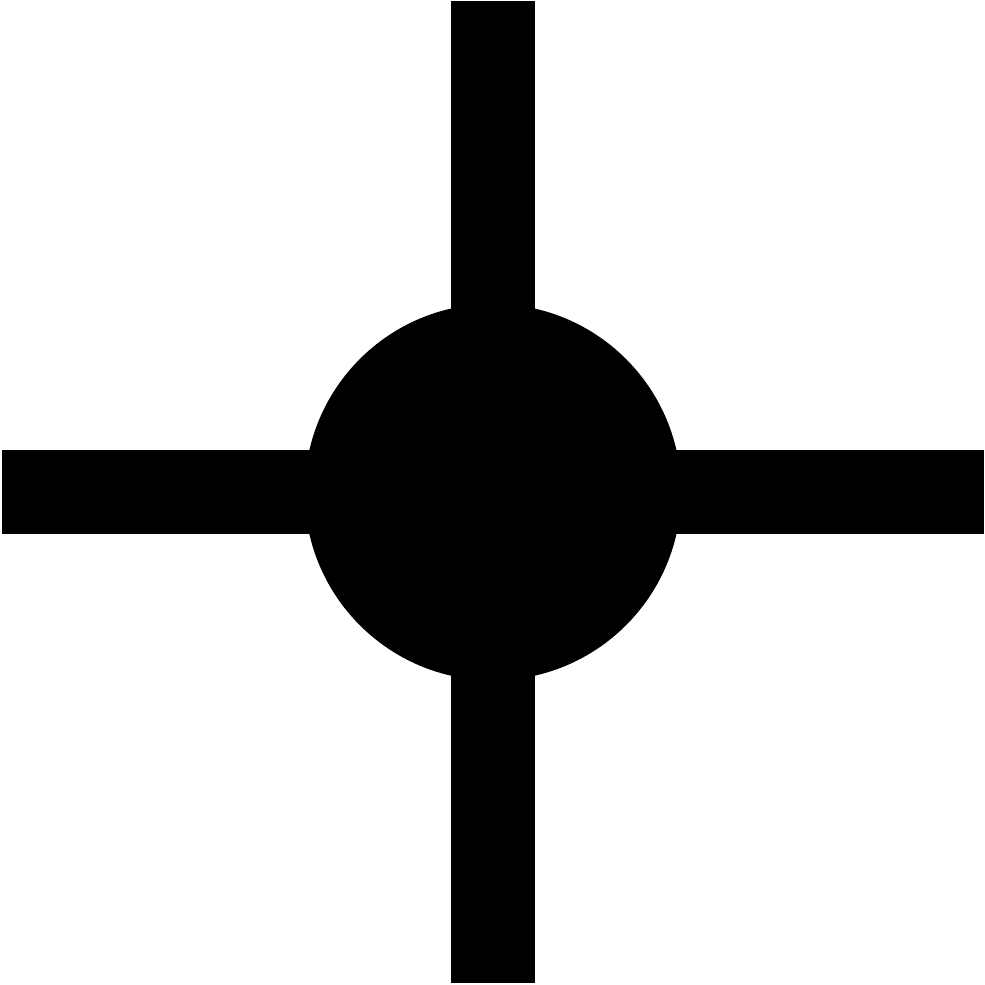}}}}\global\long\def\symNSWE{\texNSWE}

\newcommand*\texNWE{\vcenter{\hbox{\includegraphics[width=0.375cm]{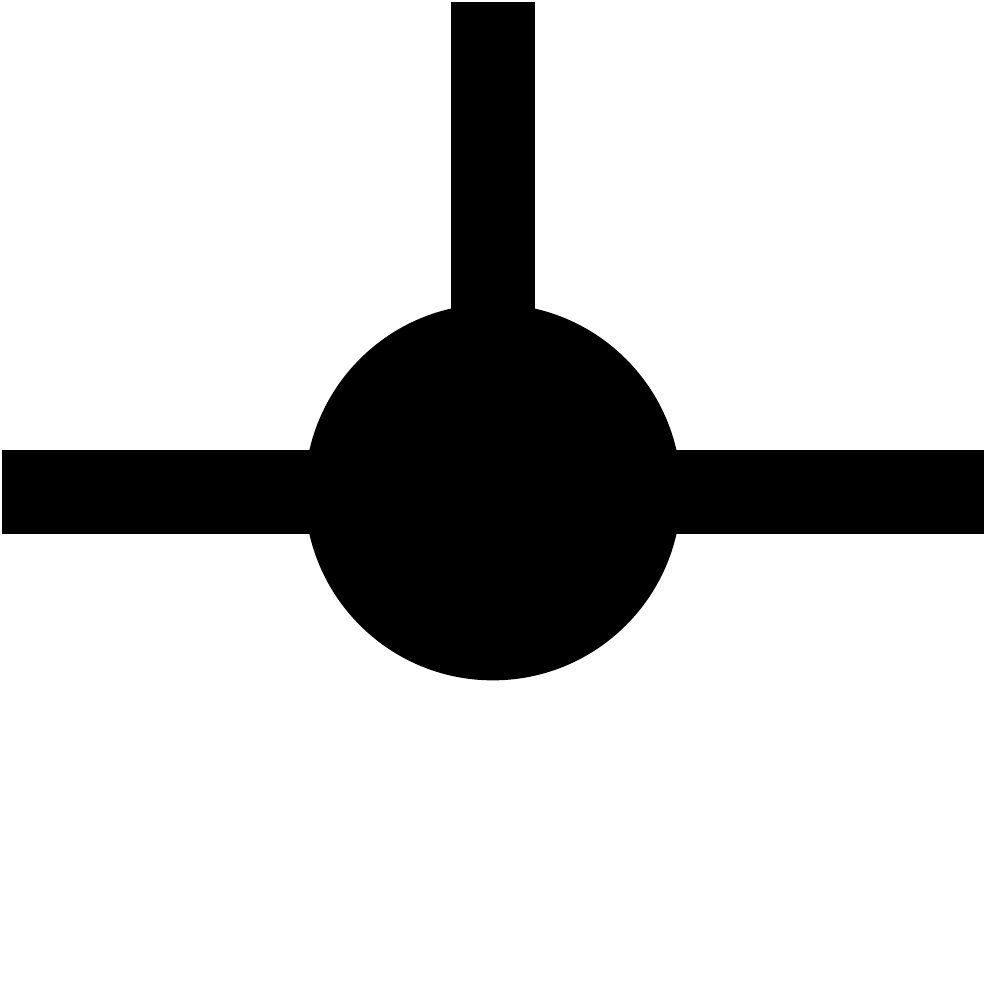}}}}\global\long\def\symNWE{\texNWE}
\newcommand*\texSWE{\vcenter{\hbox{\includegraphics[width=0.375cm]{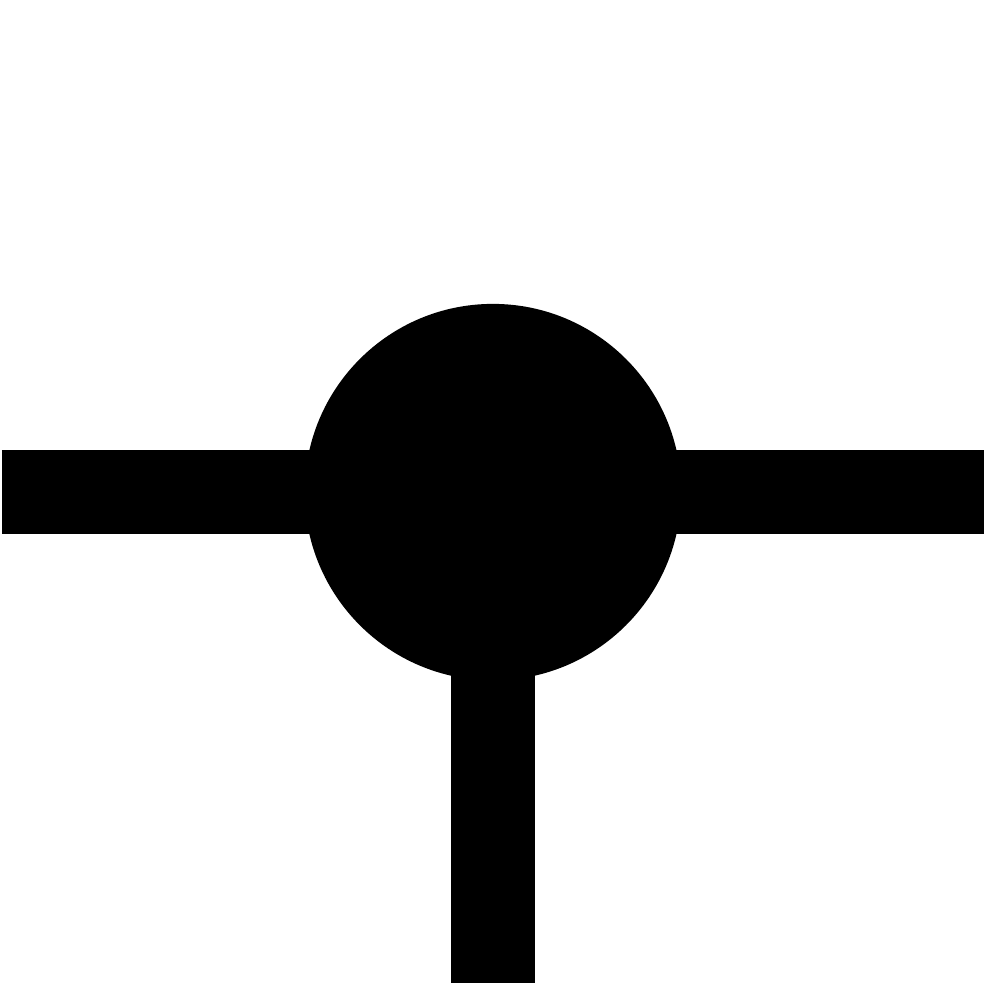}}}}\global\long\def\symSWE{\texSWE}
\newcommand*\texN{\vcenter{\hbox{\includegraphics[width=0.375cm]{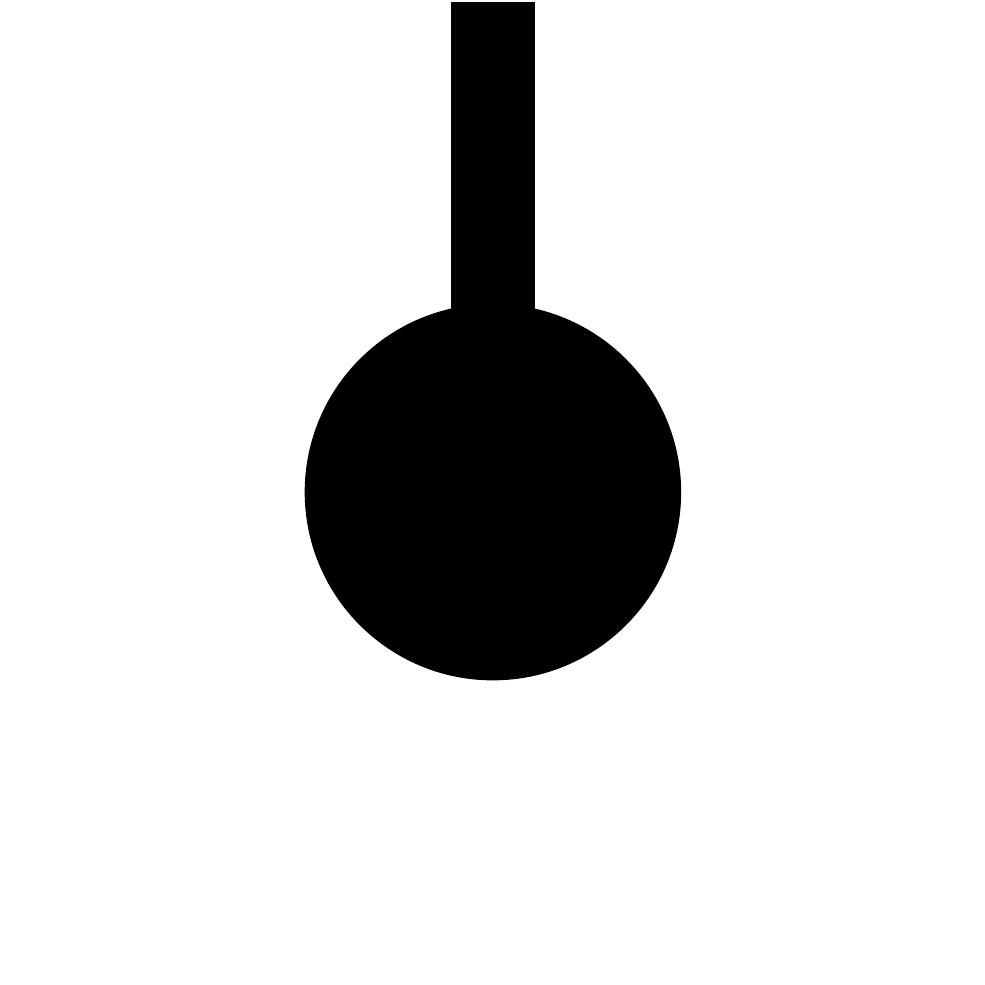}}}}\global\long\def\symN{\texN}
\newcommand*\texS{\vcenter{\hbox{\includegraphics[width=0.375cm]{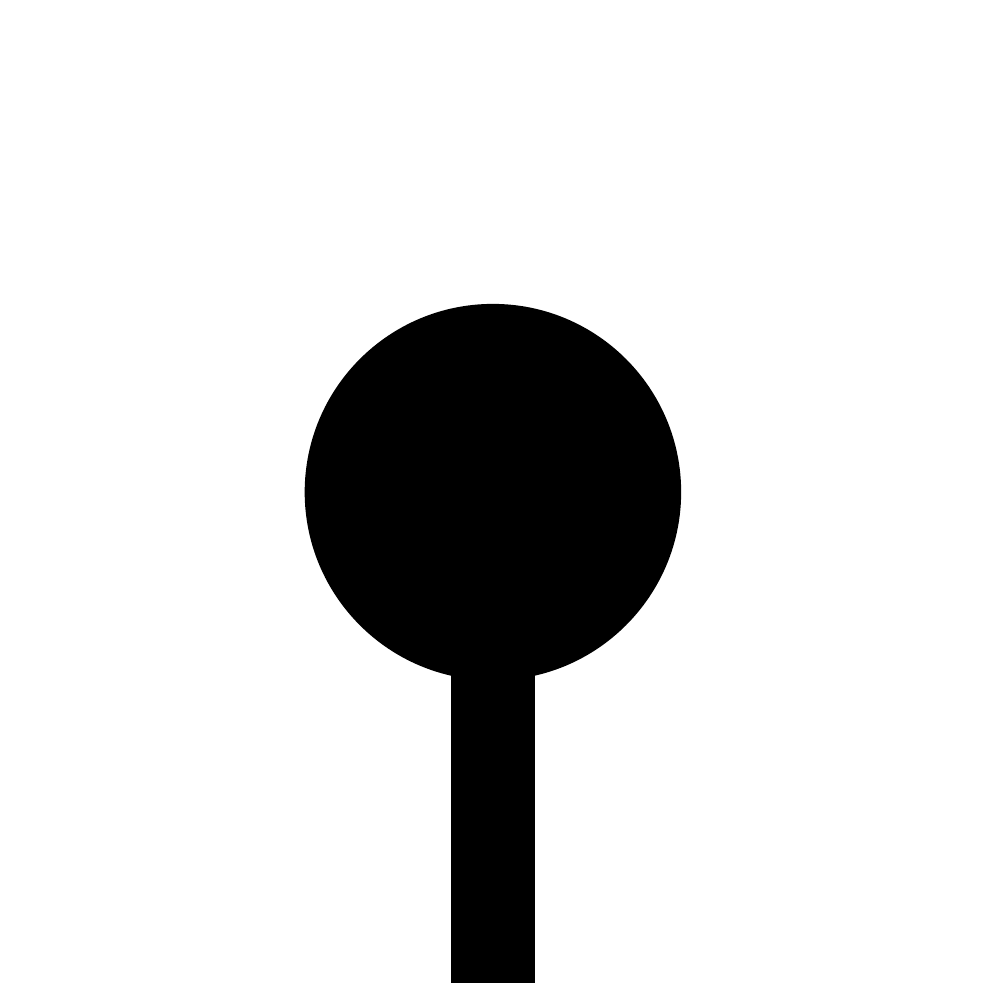}}}}\global\long\def\symS{\texS}

\title{Parameterizing the Permanent: \\
Genus, Apices, Minors, Evaluation mod $2^{k}$}

\author{Radu Curticapean\thanks{
	Simons Institute for the Theory of Computing, Berkeley, USA, and
	Institute for Computer Science and Control, Hungarian Academy of Sciences (MTA SZTAKI), Budapest, Hungary.
	Supported by ERC Starting Grant PARAMTIGHT, No. 280152.
}	
	, Mingji Xia\thanks{
	Institute of Software, Chinese Academy of Sciences, Beijing, China.
	Supported by China National 973 program 2014CB340301, China Basic Research Program (973) Grant 2014CB340302, NSFC 61003030 and NSFC 61170073.	
	}
}

\maketitle

\begin{abstract}
We identify and study relevant structural parameters for the problem
$\PerfMatch$ of counting perfect matchings in a given input graph
$G$. These generalize the well-known tractable planar
case, and they include the \emph{genus} of $G$, its \emph{apex number}
(the minimum number of vertices whose removal renders $G$
planar), and its \emph{Hadwiger number} (the size of a largest
clique minor).

To study these parameters, we first introduce the notion of \emph{combined
matchgates}, a general technique that bridges parameterized counting
problems and the theory of so-called Holants and matchgates: Using
combined matchgates, we can simulate certain non-existing gadgets
$F$ as linear combinations of $t=\mathcal{O}(1)$ existing gadgets.
If a graph $G$ features $k$ occurrences of $F$, we can then reduce
$G$ to $t^{k}$ graphs that feature only existing gadgets, thus
enabling parameterized reductions.

As applications of this technique, we simplify known $4^{g}n^{O(1)}$
time algorithms for $\PerfMatch$ on graphs of genus $g$. Orthogonally
to this, we show $\Wone[\#]$-hardness of the permanent on $k$-apex
graphs, implying its $\Wone[\#]$-hardness under the Hadwiger number.
Additionally, we rule out $n^{o(k/\log k)}$ time algorithms
under the counting exponential-time hypothesis
$\ETH[\#]$.

Finally, we use combined matchgates to prove $\Wone[\parity]$-hardness of
evaluating the permanent modulo $2^{k}$, complementing an $\mathcal{O}(n^{4k-3})$
time algorithm by Valiant and answering an open question
of Bj\"orklund. We also obtain a lower bound of $n^{\Omega(k/\log k)}$
under the parity version $\ETH[\parity]$ of the exponential-time
hypothesis.
\end{abstract}

\maketitle

\newpage
\setcounter{tocdepth}{2}
\tableofcontents

\section{Introduction}

The study of counting problems has become a classical subfield of
computational complexity since Valiant's seminal papers \cite{Valiant1979a,DBLP:journals/siamcomp/Valiant79}
introduced the class $\sharpP$ and established $\sharpP$-hardness
of counting perfect matchings in bipartite graphs. 
In particular, this proves $\sharpP$-hardness of the following generalized problem: Given a graph $G$ with edge-weights $w:E(G)\to\mathbb{Q}$,
compute the quantity
\[
\PerfMatch(G):=\sum_{\substack{M\subseteq E(G)\\
		\mathrm{perfect}\,\mathrm{matching\,of\,}G
	}
}\prod_{e\in M}w(e).
\]

In statistical physics, $\PerfMatch$ is known as the \emph{partition
	function} of the \emph{dimer model} \cite{Temperley.Fisher1961,Kasteleyn1961,Kasteleyn1967},
and the first nontrivial algorithms for the evaluation of this quantity stem from this
area. This includes the celebrated \emph{FKT method}, a polynomial-time
algorithm for computing $\PerfMatch$ on planar graphs \cite{Kasteleyn1967}.
Roughly speaking, this algorithm proceeds as follows: Given a planar
graph $G$, it constructs a \emph{Pfaffian orientation} $F$ of $G$,
which we may view as a subset $F\subseteq E(G)$ with the following
miraculous property: If we define a matrix $A$ from the adjacency
matrix of $G$ by flipping the signs of edges in $F$, then $(\PerfMatch(G))^2=\det(A)$.
Overall, this yields a reduction from planar $\PerfMatch$ to the
determinant.

In algebra and combinatorics, the quantity $\PerfMatch(G)$ for a
bipartite graph $G$ with $n+n$ vertices is better known as the \emph{permanent}
of the biadjacency matrix $A$ of $G$, defined by
\[
\perm(A)=\sum_{\substack{\sigma:[n]\to[n]\\
		\mathrm{is\:a\:permutation}
	}
}\prod_{i=1}^{n}A_{i,\sigma(i)}.
\]
The permanent is central to algebraic complexity theory, which aims
at proving the permanent to be inherently harder than the similar-looking
determinant \cite{Agrawal2006,DBLP:journals/jacm/Raz09,DBLP:journals/cc/CaiCL10}.
This would imply an algebraic analogue
of $\mathsf{P}\neq\mathsf{NP}$ \cite{Valiant1979}.

In order to obtain a more refined view on the complexity of the permanent,
and to cope with its hardness in view of practical applications, various
relaxations of this problem were studied: A celebrated randomized
\textbf{approximation} scheme \cite{Jerrum.Sinclair2004,DBLP:journals/siamcomp/JerrumS93}
allows one to approximate the permanent on matrices with non-negative
entries. Furthermore, on some \textbf{restricted graph classes}, $\PerfMatch$
can be solved in time $\O(n^{3})$: This includes the above-mentioned
planar graphs, and in fact, all graph classes of bounded genus \cite{Galluccio.Loebl,DBLP:journals/jct/Tesler00, ReggeZecchina}. (We will present more classes in the remainder of the introduction.)
As another relaxation, \textbf{modular evaluation} of the permanent was studied in Valiant's original
paper \cite{Valiant1979a}: He showed that the permanent modulo $m=2^{k}$
can be computed in time $n^{\O(k)}$ for all $k \in \mathbb{N}$, but for all $m$ containing an odd prime factor, the evaluation modulo $m$ is $\mathsf{NP}$-hard under randomized reductions.

In this paper, we consider another such refinement (and generalize existing ones) by investigating
the permanent in the framework of \textbf{parameterized complexity}.
This area was initiated by Downey and Fellows \cite{Downey.Fellows1995,Downey.Fellows1999}
and was adapted to counting problems by Flum and Grohe \cite{Flum.Grohe2004}
and McCartin \cite{McCartin2006}. In parameterized counting complexity,
the objects in study are counting problems that come with \emph{parameterizations}
$\pi:\{0,1\}^{*}\to\mathbb{N}$, and a central question is whether
such problems are \emph{fixed-parameter tractable (fpt)}. A given
problem is fpt if it can be solved in time $f(\pi(x))|x|^{\O(1)}$
on input $x$, for a computable function $f$ that depends only on
the parameter value, but not on $|x|$.
If we fail to find an fpt-algorithm for a given parameterized problem, we can often give evidence
that no such algorithm exists by proving its $\Wone[\#]$-hardness,
the parameterized analogue of $\sharpP$-hardness. 
(A more detailed exposition can be found in Section~\ref{sec:General-preliminaries}.)

By studying natural parameterizations $\pi$ of the input, we obtain
a fine-grained complexity analysis that could not be achieved by considering
the input size $|x|$ alone. For instance, consider the problem $\mathrm{VertexCover}$,
which asks whether a graph $G$ on $n$ vertices admits a vertex-cover
of size $k$. This problem is $\mathsf{NP}$-complete, but it can
be solved in time $n^{\O(k)}$ for every fixed $k$, and it is actually
even fpt in the parameter $k$, as we can find \cite{Downey.Fellows1995}
and even count \cite{Flum.Grohe2006} vertex-covers of size $k$ in
time $2^{k}n^{\O(1)}$. On the other hand, we can decide in polynomial
time whether $G$ contains a matching of size $k$, but the problem
of counting $k$-matchings is $\sharpP$-complete, and in fact even
$\Wone[\#]$-complete when parameterized by $k$ \cite{Curticape2013,Curtican.Marx2014}.

\subsection{Genus, apices and excluded minors}

To investigate the parameterized complexity of the permanent, we first
identify interesting parameterizations for this problem. For instance,
the maximum degree $\Delta(G)$ of the input graph $G$ is not particularly
interesting, since the permanent is already $\sharpP$-complete on
$3$-regular graphs \cite{Dagum.Luby1992}. That is, even an $n^{f(\Delta(G))}$
time algorithm for some function $f$ (and an fpt-algorithm in particular) would imply $\mathsf{P}=\sharpP$.
However, it turns out that the known polynomial-time solvable graph
classes for $\PerfMatch$ point us towards a natural parameter, namely
the size of a smallest excluded minor. Here, a minor $H$ of a graph $G$ is a graph that can be obtained from $G$ by deletions of edges
and/or vertices, and contraction of edges. To explain the significance of minors for counting perfect matchings, we first survey the known algorithms for $\PerfMatch$, all of which can be considered as generalizations of the FKT method for planar graphs.

\begin{description}
	
	\item[Excluding $K_{3,3}$ or $K_5$:]
	It was shown by Little \cite{Little1974} and later by Vazirani \cite{Vazirani1989} (who gave a parallelized algorithm)
	that $\PerfMatch$ can be solved in time $\O(n^{3})$ on graphs excluding
	the minor $K_{3,3}$. A similar result was recently shown by Straub
	et al.~\cite{Straub.Thierauf2014} for graphs excluding $K_{5}$.
	Note that the FKT method gives an $\O(n^{3})$ time algorithm on graphs
	excluding \emph{both} $K_{3,3}$ and $K_{5}$, whereas the two above
	algorithms show that excluding \emph{either} minor entails the polynomial-time
	solvability of $\PerfMatch$. For the $K_{3,3}$-free case, this was
	shown by constructing a Pfaffian orientation. The $K_{5}$-free case
	was shown by a different technique; in particular, $K_{5}$-free graphs
	do not necessarily admit Pfaffian orientations.
	
	\item[Excluding single-crossing minors:]
	Extending the above item,
	it was recently shown by Curticapean \cite{Curticape2014} that $\PerfMatch$
	can be solved in time $\O(n^{4})$ on any class excluding a fixed
	\emph{single-crossing minor} $H$, i.e., a minor that can be drawn
	in the plane with at most one crossing, such as $K_{3,3}$ or $K_{5}$.
	In fact, it is shown that $\PerfMatch$ is fpt in the size of the
	smallest excluded single-crossing minor. This algorithm does not inherently
	rely upon Pfaffian orientations, apart from a black-box algorithm
	for planar $\PerfMatch$.
	
	\item[Bounded-genus graphs:]
	Another line of extensions of the FKT
	method is to graphs of bounded \emph{genus}: It was shown independently by Gallucio
	and Loebl \cite{Galluccio.Loebl}, Tesler \cite{DBLP:journals/jct/Tesler00}
	and Regge and Zechina \cite{ReggeZecchina} that $\PerfMatch$ can
	be solved in time $\O(4^{g}n^{3})$ on $n$-vertex graphs $G$ of
	genus $g$. In the framework of fixed-parameter tractability, this can be read as $\PerfMatch$ being fpt when parameterized by the genus
	of $G$. The algorithms for the bounded-genus case proceed by expressing $\PerfMatch(G)$ as
	the linear combination of $4^{g}$ determinants derived from Pfaffian
	orientations. In the present paper, we give an alternative proof of
	this theorem that proceeds by reduction to $4^{g}$ instances of planar
	$\PerfMatch$. Together with the previous item, this eliminates the
	need for Pfaffian orientations from all known algorithms for $\PerfMatch$
	except for the planar case.
	
\end{description}

From the above list, we can draw the conclusion that every\emph{ known }polynomial-time
solvable graph class for $\PerfMatch$ excludes some fixed minor.\footnote{This statement comes with a caveat: For instance, we can determine the number of perfect matchings in a complete graph in polynomial time by means of a closed formula. The class of complete graphs clearly excludes no fixed minor. However, we cannot solve the (weighted) problem $\PerfMatch$ on this class in polynomial time, as edge-weights would allow us to simulate arbitrary graphs, for which counting perfect matchings is $\sharpP$-complete.}
This is clear for the first two items, and furthermore, the graphs
of genus $g\in\N$ are easily seen to exclude a complete graph of size $\mathcal{O}(g)$.
Since this shows that excluded minors have been a driving force behind
polynomial-time algorithms for $\PerfMatch$, it is natural to study
this problem under the more general \emph{Hadwiger number}
\[
\hadw(G):=\max\{k\in\N:\ G\mbox{ contains a \ensuremath{K_{k}}-minor}\}.
\]

Note that planar graphs have Hadwiger number at most $4$. More generally,
if the genus of $G$ or the size of the smallest excluded single-crossing
minor is bounded, then $\hadw(G)$ is bounded as well, but the converse
does not hold. However, the \emph{Graph Structure Theorem} \cite{Robertson.Seymour2003},
a celebrated result in graph minor theory \cite{Robertson.Seymour2004},
yields a decomposition of the graphs with fixed Hadwiger number $k$ into
graphs that have genus $c=c(k)$ except for $c$ occurrences of certain
defects, namely so-called vortices and apices.
Such decompositions have proven immensely useful for fpt-algorithms
on graphs excluding fixed minors, see \cite{Marx.Pilipczuk2014,Demaine.Hajiaghayi2009,Demaine.Hajiaghayi2005a,Demaine.Hajiaghayi2005,Demaine.Hajiaghayi2002,Fomin.Demaine2015}.
If a problem can be solved efficiently on planar instances and we
can extend this to bounded-genus instances, as in the case of $\PerfMatch$,
then with a leap of faith, the Graph Structure Theorem allows us to
hope for an fpt-algorithm under the more general parameterization by Hadwiger number.
Our following negative result however shatters these hopes for the case of $\PerfMatch$.
\begin{thm}
	\label{thm: hadw hard}The zero-one permanent is $\Wone[\#]$-hard when parameterized
	by the Hadwiger number. In other words, computing $\PerfMatch$ is $\Wone[\#]$-hard when parameterized by the Hadwiger number, even on bipartite graphs without edge-weights.
\end{thm}
We show this by proving the following stronger statement: Let us define
the apex number
\[
\apex(G):=\min\{k\in\N\mid\exists S\subseteq V(G)\mbox{ of size }k:\ G-S\mbox{ is planar}\}.
\]

This parameter, studied in \cite{Marx.Schlotter2012}, measures the
distance of a graph to planarity by vertex deletions. Note that planar
graphs have apex number $0$. Using the apex number as parameter,
we can generalize planar graphs in a way that is orthogonal to the
genus parameter: There are graphs on which any one of these parameters
is bounded, while the other is not. However, it can be verified that
$\hadw(G)\leq\O(\apex(G))$. This allows us to obtain Theorem~\ref{thm: hadw hard}
as a corollary from the following result, which we consider to be
of independent interest.
\begin{thm}
	\label{thm: apex hard}The permanent is $\Wone[\#]$-hard when parameterized
	by the apex number. Assuming the exponential-time hypothesis $\ETH[\#]$,
	it admits no $n^{o(k/\log k)}$ time algorithm on $k$-apex graphs
	with $n$ vertices.
\end{thm}
This contrasts with the fpt-algorithm for $\PerfMatch$ when parameterized
by genus. We observe that $\PerfMatch$ can be computed
easily in time $n^{k+\O(1)}$ on $k$-apex graphs by means of brute-force,
so the lower bound under $\ETH[\#]$ is almost tight.
However, it should be noted that no similar algorithm is known for the Hadwiger
number: At least to us, it remains an important open question whether $\PerfMatch$
can be solved in time $n^{f(k)}$ on graphs excluding the complete
graph $K_{k}$ as minor.

\subsection{Evaluating the permanent modulo $2^{k}$}

In the following, we depart from structural parameters of the input
graph $G$ and consider the evaluation of the permanent modulo $2^{k}$.
In the seminal paper \cite{Valiant1979a}, not only did Valiant prove
$\sharpP$-completeness of the permanent, but he also studied the
complexity of evaluating the permanent modulo fixed numbers $m\in\mathbb{N}$.

Observe that $\perm(A)$ and $\det(A)$ are equivalent modulo $2$ for any matrix $A$,
giving a polynomial-time algorithm for the permanent modulo $2$.
On the other hand, for odd primes $p$, Valiant's original proof shows
that the permanent modulo $p$ is $\mathsf{Mod}_{p}\mathsf{P}$-complete.
That is, we can reduce counting satisfying assignments to $3$-CNF
formulas modulo $p$ to the permanent modulo $p$. This also shows
the $\mathsf{NP}$-hardness of the latter problem under randomized reductions, and this
holds more generally whenever the modulus $m$ contains an odd prime factor, that is, when $m$ is not a power of two.

For the remaining cases $m=2^{k}$ however, Valiant
\cite{Valiant1979a} showed an $\O(n^{4k})$ time algorithm for
evaluating the permanent modulo $2^{k}$ on $n$-vertex graphs, which
was recently improved to $n^{k+\mathcal{O}(1)}$ time by Bj\"orklund, Husfeldt and Lyckberg~\cite{BHL_PermMod}.
Given these results, it is natural to study this problem in the framework
of parameterized complexity, thus asking whether we can compute the
permanent modulo $2^{k}$ in time $n^{o(k)}$ or even $f(k)n^{\O(1)}$.
This was also posed as an open problem in \cite{BHL_PermMod}.
Please recall that this question is indeed only interesting for $m=2^{k}$:
As stated in the previous paragraph, on all other \emph{fixed} $m\in\N$,
the problem is $\mathsf{NP}$-hard under randomized reductions.

We rule out the fixed-parameter tractability 
of the permanent modulo $2^k$
by the following stronger
hardness result, which also establishes an unexpected connection
to the apex parameter introduced before:
Evaluating the permanent modulo $2^{k}$ on $k$-apex graphs is $\Wone[\parity]$-hard,
that is, an fpt-algorithm for this problem would imply one for counting
$k$-cliques modulo $2$, a problem that was shown to be $\Wone$-hard under randomized
reductions by a recent result of Bj\"orkund, Dell and Husfeldt~\cite{DBLP:conf/icalp/BjorklundDH15}.
We also obtain an almost-tight lower bound under $\ETH[\parity]$,
the parity version of the exponential-time hypothesis $\ETH$.
\begin{thm}
	\label{thm: perm mod 2^k}The evaluation of the permanent modulo $2^{k}$
	is $\Wone[\parity]$-hard when parameterized by $k$, even when restricted
	to $k$-apex graphs. Assuming $\ETH[\parity]$, there is no $n^{o(k/\log k)}$
	time algorithm for this problem.
\end{thm}
Theorem~\ref{thm: perm mod 2^k} is proven by reduction from the following
problem $\pPart[\parity]$: Given vertex-colored graphs $H$ and $G$
as input, where each color in $H$ appears exactly once, count modulo $2$ the subgraphs
of $G$ that are isomorphic to $H$, respecting colors. It was shown
that the decision version of this problem, which is $\Wone$-hard, can be reduced to $\pPart[\parity]$
by means of randomized reductions \cite{DBLP:conf/icalp/BjorklundDH15}.
Furthermore, assuming $\ETH[\parity]$, an argument by Marx \cite{Marx2010}
implies that $\pPart[\parity]$ cannot be solved in time $n^{o(\ell/\log\ell)}$
for $\ell$-edge graphs $H$ and $n$-vertex graphs $G$.

In our reduction, we transform a given instance $(H,G)$ for $\pPart[\parity]$
with an $\ell$-edge graph $H$ to $3^{\ell}$ instances of the permanent
modulo $2^{2\ell+1}$ on $2\ell$-apex graphs with $\O(\ell^{2}n^{2})$
vertices. Thus, if we can prove better lower bounds for finding $k$-edge
subgraphs, then those bounds carry over to the seemingly unrelated
problem of evaluating permanents modulo $2^{k}$, even on $k$-apex graphs.
On the other hand, a randomized $n^{o(k)}$ time algorithm for the
permanent modulo $2^{k}$ on $k$-apex graphs would imply one for
$\pPart$ on $k$-edge graphs $H$, thus falsifying a hypothesis posed
by Marx \cite{Marx2010}.

\subsection{Proof technique: Linear combinations of signatures}

We phrase our proofs in the language of so-called Holant problems
\cite{Cai.Lu2007} and matchgates \cite{Cai.Lu2007,Cai.Choudhary2006,Cai.Gorenstein2013}.
Please consider Section~\ref{sec:Holants-and-linear}
for a more detailed introduction into these topics. In our proofs, we reformulate
parameterized counting problems as Holant problems (specific weighted
sums over assignments to the edges of graphs) and then try to realize the occurring signatures (local constraints at vertices) by certain
matchgates (gadgets). However, many required signatures cannot be
realized by matchgates. The key idea in this paper
is that such unrealizable signatures can sometimes still be realized
as \emph{linear combinations }of matchgate signatures.

To this end, we proceed as follows: First, we show how to simulate
non-existing gadgets $F$ as a formal linear combination of realizable
gadgets $F_{1},\ldots,F_{t}$, typically with $t=\O(1)$. Then, if
a graph $G$ features $k$ occurrences of $F$, we can easily reduce
$G$ to $t^{k}$ graphs that feature only occurrences of $F_{1},\ldots,F_{t}$.
Each of these $t^{k}$ graphs can then be handled by an algorithm
(when we aim at positive results) or by an oracle call (when proving
hardness results). The generality of our technique allows it to be
applied to various parameterized problems. For instance, a recent
$\Wone[\#]$-hardness proof for counting $k$-matchings \cite{Curtican.Marx2014}
can also be rephrased in this framework.

As pointed out by Tyson Williams, a similar idea appears under the notion of \emph{vanishing signatures} \cite{Guo.Lu2013,DBLP:conf/stoc/CaiGW13}.
These however apply linear combinations in a quite different setting.
In particular, they consider no connections to parameterized complexity.

\subsection*{Organization of the paper}

In Section \ref{sec:General-preliminaries}, we introduce notions
from parameterized complexity, exponential-time complexity, and we
prove $\Wone[\#]$-hardness of a modified version of the problem $\pGrid[\#]$,
our main reduction source for subsequent hardness proofs. In Section
\ref{sec:Holants-and-linear}, we introduce Holant problems and matchgates,
including some particular matchgates required in later sections. We
also introduce our proof technique of linearly combined signatures. This finishes the general introduction of our proof techniques.

In Section \ref{sec: genus}, we then give a first application of the machinery developed in the previous sections
by proving a $4^{g} \cdot n^{\O(1)}$ time algorithm for $\PerfMatch$ on
graphs of genus $g$. In Section \ref{sec: permanent k-apex}, we
then prove Theorem~\ref{thm: apex hard}, which asserts $\Wone[\#]$-hardness
of $\PerfMatch$ on bipartite unweighted $k$-apex graphs and implies Theorem \ref{thm: hadw hard}, the hardness under the Hadwiger number parameter. In Section
\ref{sec: permanent modulo}, we introduce a more involved construction
and an additional technique called \emph{discrete derivatives} to
transform the argument from Section~\ref{sec: permanent k-apex}
to a proof of Theorem~\ref{thm: perm mod 2^k}.

\section{\label{sec:General-preliminaries}General preliminaries}

For $n\in\N$, we write $[n]:=\{1,\ldots,n\}$. The graphs $G$ in
this paper are undirected, but they may feature parallel edges and
edge-weights. All \emph{hardness results} are however shown for \emph{simple}
graphs featuring no parallel edges and no edge-weights. We write $uv\in E(G)$
for an edge of $G$, and given $v\in V(G)$, we denote the edges incident
with $v$ by $I(v)$. Sometimes, we consider graphs to be embedded
on surfaces, see \cite{Diestel2012}.

For numbers $n\in\mathbb{N}$, we abbreviate $\parity n:=(n\ \mod\ 2)$.
Given a bitstring $x\in\{0,1\}^{*}$, we write $\hw(x):=\sum_{i}x_{i}$
for its \emph{Hamming weight}, and we define 
\begin{eqnarray*}
	\odd(x)  & := & \parity\hw(x), \\
	\even(x) & := & 1-\parity\hw(x).
\end{eqnarray*}
We write $\supp(f)$ for the support of a function $f$. For predicates
$\varphi$, we use the Iverson bracket notation 
\[
[\varphi]:=\begin{cases}
1 & \varphi\mbox{ is true,}\\
0 & \mbox{otherwise.}
\end{cases}
\]

Let $A$ and $B$ be sets; we define certain abbreviations for subsets of $A\times B$. 
For $b\in B$, we write $(\star,b)=\{(a,b)\mid a\in A\}$
for the \emph{column} at $b$. For $a\in A$, we write $(a,\star)=\{(a,b)\mid b\in B\}$
for the \emph{row} at $a$. 
We use this notation only when $A$ and $B$ are unambiguous from the context.
For $k\in\mathbb{N}$, we say that $(i,j)\in[k]^{2}$ and $(i',j')\in[k]^{2}$
are \emph{vertically adjacent} if $|i-i'|=1$ and $j=j'$. Likewise,
we call such pairs \emph{horizontally adjacent} if $|j-j'|=1$ and
$i=i'$.

\subsection{Parameterized complexity}

Parameterized counting problems are problems $\mathrm{A}/\pi$, where
$\mathrm{A}:\{0,1\}^{*}\to\mathbb{C}$ is a counting problem and $\pi:\{0,1\}^{*}\to\N$
is a polynomial-time computable parameterization, see \cite{Flum.Grohe2004}.
We define $\FPT$ as the class of all problems $\mathrm{A}/\pi$ such that $\mathrm{A}$ can be solved in time $f(\pi(x))|x|^{\O(1)}$.
Likewise, we define $\XP$ as the class of problems $\mathrm{A}/\pi$ that can
be solved in time $|x|^{f(\pi(x))}$, where $f:\N\to\N$ is a computable
function. In the following, we define the classes $\Wone$, $\Wone[\#]$
and $\Wone[\parity]$ we referred to in the introduction, using the following reduction
notions.
\begin{defn}
	[\cite{Flum.Grohe2004}]\label{def: fpt-reductions}Let $\mathrm{A}/\pi$
	and $\mathrm{B}/\pi'$ be parameterized counting problems. 
	\begin{itemize}
		\item We call $f:\{0,1\}^{*}\to\{0,1\}^{*}$ a \emph{parsimonious fpt-reduction}
		and write $\mathrm{A}/\pi \leqFptPar \mathrm{B}/\pi'$ if there
		are computable functions $r,s$ such that the following holds for all
		$x\in\{0,1\}^{*}$: 
		\begin{enumerate}
			\item We have $\mathrm{A}(x)=\mathrm{B}(f(x))$. 
			\item The running time of $f$ is bounded by $r(\pi(x))\cdot |x|^{\mathcal{O}(1)}$. 
			\item We have $\pi'(f(x))\leq s(\pi(x))$.
		\end{enumerate}
		If $\mathrm{A}$ and $\mathrm{B}$ are decision problems, 
		replace the first condition by ``$x\in\mathrm{A}$ iff $f(x)\in\mathrm{B}$'',
		and write $\mathrm{A}/ \pi \leqFpt \mathrm{B} / \pi'$.
		\item We call an algorithm $\mathbb{T}$ a \emph{Turing fpt-reduction}
		and write $\mathrm{A}/\pi \leqFptT \mathrm{B}/\pi'$ if there are
		computable functions $r$ and $s$ such that the following holds for
		all $x\in\{0,1\}^{*}$: Firstly, the running time of $\mathbb{T}$ on $x$ is bounded by $r(\pi(x))|x|^{\mathcal{O}(1)}$. 
		Secondly, every oracle query $y$ issued by $\mathbb{T}$ on $x$ satisfies
		$\pi'(y)\leq s(\pi(x))$.
	\end{itemize}
\end{defn}
We introduce $\Wone$, $\Wone[\parity]$
and $\Wone[\#]$ as the closures of clique-related problems
under fpt-reductions.
\begin{defn}
	Consider the following parameterized problems and complexity classes:
	\begin{itemize}
		\item Let $\pClique/k$ denote the problem of \emph{deciding}, on input
		a graph $G$ and $k\in\N$, whether $G$ contains a $k$-clique. 
		Let $\Wone$ denote the set of all problems $\mathrm{A}/\pi$ with $\mathrm{A}/\pi \leqFpt \pClique/k$.
		\item Let $\pClique[\#]/k$ denote the problem of determining, on input
		$G$ and $k$, the \emph{number} of $k$-cliques in $G$. 
		Let $\Wone[\#]$ denote the set of all problems $\mathrm{A}/\pi$ with $\mathrm{A}/\pi\leqFptPar\pClique[\#]/k$.
		\item Let $\pClique[\parity]/k$ denote the problem of \emph{deciding},
		on input $G$ and $k$, whether $G$ contains an \emph{odd} number
		of $k$-cliques. 
		Let $\Wone[\parity]$ denote the set of all $\mathrm{A}/\pi$
		with $\mathrm{A}/\pi\leqFpt\pClique[\parity]/k$.
	\end{itemize}
\end{defn}
It is a standard assumption of parameterized complexity theory that
$\FPT\neq\Wone$ holds, implying $\FPT\neq\Wone[\#]$. The problem
$\pClique/k$ is $\Wone$-complete by definition, so this assumption
can equivalently be considered as the statement that $\pClique/k$ is not fixed-parameter
tractable. Furthermore, it has been recently shown in \cite[Theorem 5]{DBLP:conf/icalp/BjorklundDH15}
that $\pClique[\parity]/k$ is $\Wone$-hard under randomized parameterized
reductions with constant one-sided error. Therefore, an fpt-algorithm
for $\pClique[\parity]/k$ would imply a randomized fpt-algorithm
for $\pClique/k$, which is considered almost as unlikely as $\FPT=\Wone.$

\subsection{Exponential-time complexity}

We also consider conditional lower bounds on the running times required
to solve problems. These are based on different exponential-time hypotheses,
introduced by \cite{Impagliazzo.Paturi2001,Impagliazzo.Paturi2001a}
and \cite{Dell.Husfeldt2014}.
\begin{defn}
	The exponential-time hypothesis $\ETH$, introduced in \cite{Impagliazzo.Paturi2001,Impagliazzo.Paturi2001a},
	claims that the satisfiability of $3$-CNF formulas on $n$ variables 
	cannot be decided in time $2^{o(n)}n^{O(1)}$.
	The hypothesis $\ETH[\#]$ postulates the same lower bound for counting the number of satisfying assignments to
	$3$-CNF formulas, and $\ETH[\parity]$ postulates the same for computing
	the parity of the number of satisfying assignments.
\end{defn}
The hypothesis $\ETH$ implies a lower bound for $\pClique/k$, and
thus also $\FPT\neq\Wone$: It was shown in \cite{Chen.Chor2005,Chen.Huang2006}
that $\pClique/k$ cannot be solved in time $f(k)\cdot n^{o(k)}$ on $n$-vertex
graphs, for any computable function $f$. Furthermore, if a problem $\mathrm{A}/\pi$ cannot be solved in time $f(k)\cdot n^{g(k)}$ under $\ETH$, and we can reduce $\mathrm{A}/\pi$
to $\mathrm{B}/\pi'$ with a reduction $f$ that satisfies $\pi'(f(x))\in\O(\pi(x))$
for all $x$, then it can be seen that $\mathrm{B}/\pi'$ cannot be solved in time
$f'(k)\cdot n^{\Omega(g(k))}$ under $\ETH$, for any computable function $f'$.

By an isolation argument similar to the Valiant-Vazirani theorem \cite{Valiant.Vazirani1986},
it was shown in \cite{CalabroIKP08} that a $2^{o(n)}$ time algorithm
for counting satisfying assignments to $3$-CNF formulas modulo $2$
implies a randomized $2^{o(n)}$ time algorithm for deciding the existence
of a satisfying assignment. In other words, a randomized version $\ETH[r]$
of $\ETH$ implies $\ETH[\parity]$; see also \cite{Dell.Husfeldt2014} for more details.

\subsection{Grid tilings and vertex-colored subgraphs}

We will reduce from the problem $\pGrid$ of deciding the existence of grid tilings, as well as its counting version  $\mathsf{\#}\pGrid$ and its parity counting version $\pGrid[\parity]$. The decision version $\pGrid$ was introduced by
Marx \cite{Marx2012} in order to obtain lower bounds for planar multiway
cut, but grid tilings have since proven to be generally useful for proving hardness of problems on planar structures \cite{Marx.Pilipczuk2014}. 
\begin{defn}
	\label{def: GridTiling}The inputs to the problem $\pGrid$ are numbers
	$n,k\in\mathbb{N}$, together with a set $\mathcal{C}\subseteq[k]^{2}$
	and a function {\small{}$\mathcal{T} : \C \to 2^{[n]^{2}} $}. 
	The task is to decide whether there exists
	a\emph{ grid tiling} of $\mathcal{T}$, i.e., a function $a:[k]^{2}\to[n]^{2}$
	such that:
	\begin{enumerate}
		\item For horizontally adjacent $\kappa,\kappa'\in[k]^{2}$, the first components
		of $a(\kappa)$ and $a(\kappa')$ agree.
		\item For vertically adjacent $\kappa,\kappa'\in[k]^{2}$, the second components
		of $a(\kappa)$ and $a(\kappa')$ agree. 
		\item For all $\kappa\in\mathcal{C}$, we have $a(\kappa)\in\mathcal{T}(\kappa)$.
	\end{enumerate}
	On the same inputs, we also define the problem $\mathsf{\#}\pGrid$,
	which asks to determine the \emph{number} of grid tilings, and the
	problem $\pGrid[\parity]$, which asks to determine the \emph{parity}
	of this number. All three problems are parameterized by $k$.\end{defn}
It should be noted that our definition of $\pGrid$ is actually a generalized version of Marx's
original formulation \cite{Marx2012}: In his definition, the set $\C$ of any instance is fixed to
$\C=[k]^{2}$. That is, the third condition of Definition~\ref{def: GridTiling}
is required to apply for \emph{all} $\kappa\in[k]^{2}$, whereas in our formulation, only a subset is relevant.
In particular, we may choose sparse subsets $\C$ with $|\C| = \mathcal{O}(k)$, which will make the generalized grid tiling problems very useful in proving lower bounds under the exponential-time hypotheses.

By reduction from $k$-cliques, Marx showed that $\pGrid$ is complete
for $\Wone$. A simple adaptation of this reduction shows that the same holds for its counting and parity
version, where $\Wone[\#]$ and $\Wone[\parity]$ take the part of
$\Wone$. In the remainder of this subsection, we give a different reduction, which chooses partitioned
subgraph isomorphisms rather than $k$-cliques as a reduction source.
This allows us to transfer an almost-tight conditional lower bound
for subgraph isomorphisms under $\ETH$ to $\pGrid$.
\begin{defn}
	For $k\in\N$, a \emph{$[k]$-colored graph} is a pair $(H,c)$, where
	$H$ is a graph and $c:V(H)\to[k]$ is a coloring. We call $(H,c)$ \emph{colorful}
	if $c$ is bijective. This implies of course that $H$ has $k$
	vertices. 
	
	For $[k]$-colored graphs $(H,c)$ and $(G,c')$, we say
	that $(H,c)$ is \emph{color-preserving isomorphic} to $(G,c')$ if
	there exists an isomorphism $f$ from $H$ to $G$ such that $c(v)=c'(f(v))$
	holds for all $v\in V(H)$. To simplify notation, we will often write
	$G$ rather than $(G,c)$ for a colored graph.
	
	The problem $\pPart$ is defined as follows: The input consists of $[k]$-colored
	graphs $H$ and $G$, where $H$ is colorful. The task is to decide
	whether there exists a copy of $H$ in $G$, which is a (not necessarily
	induced) subgraph $F$ of $G$ such that $H$ and $F$ are color-preserving
	isomorphic. Likewise, the problem $\pPart[\#]$ asks to determine
	the \emph{number} of copies of $H$ in $G$, and $\pPart[\oplus]$
	asks to determine its parity. All problems are parameterized by $k$.
\end{defn}
It can be shown by a parsimonious reduction from $\pClique$ that
the problem $\pPart$ is $\Wone$-complete, and this implies similar
statements for its other variants as well. We omit the elementary proof.
\begin{lem}
	\label{lem: PartitionedSub hard}The three variants of $\pPart$ are
	complete for $\Wone$, $\Wone[\#]$ or $\Wone[\parity]$, respectively.\end{lem}
\begin{rem}
	\label{rem: color-preserve-preprocessing}Let $H$ and $G$ be $[k]$-colored
	such that $H$ is colorful; we assume $V(H)=[k]$ without limitation
	of generality. If $F$ is a $H$-copy in $G$ and $uv\in E(F)$ is
	an edge with endpoints of colors $i$ and $j$ for some $i,j\in[k]$,
	then the edge $ij$ must be present in $H$. 
	
	We may therefore assume the following: Whenever an instance $(H,G)$ to $\pPart$ is given, then
	for all $i,j\in[k]$ with $ij\notin E(H)$, the graph $G$ contains no
	edges between $i$-colored and $j$-colored vertices. Otherwise, we can delete these edges without affecting the set of color-preserving $H$-copies in $G$.
\end{rem}
In the following, we consider instances $(H,G)$ for $\pPart$ with
$n=|V(G)|$ and $k=|V(H)|$. We can solve each such instance in time
$n^{\O(k)}$ by brute force, and by reduction from $\pClique$, it
was shown that algorithms with running time $f(k) \cdot n^{o(k)}$ would refute
$\ETH$, see \cite{Chen.Chor2005,Chen.Huang2006}. 

This lower bound alone would however not suffice for our purposes
of proving tight lower bounds: In the reductions from $\pPart$ to
the permanent we construct later, each \emph{edge} of $H$ will incur
some constant parameter blowup. As an example, on instances $(H,G)$,
our reduction images for the permanent will have $\O(|E(H)|)$ apices,
which amounts to $\O(k^{2})$ apices if $H$ is a $k$-clique. Thus,
if we reduced from $\pClique$ for our lower bound, then $\ETH$ could
only rule out algorithms with running time $n^{o(\sqrt{t})}$ for
the permanent on $t$-apex graphs. This is however obviously far from the upper bound of $\O(n^{t+3})$ time obtained by the brute-force algorithm, and we would not consider such a result to be satisfactory. 

To avoid this problem, we use a refined lower bound for $\pPart$, shown
also by Marx, which allows to assume that $H$ has constant degree,
and thus, only $\O(k)$ edges, see \cite[Corollary~6.3]{Marx2010}.
\begin{thm}[\cite{Marx2010}]
	\label{thm: PartitionedSub Marx}Assuming $\ETH$, there is a universal
	constant $C^{*}$ such that $\pPart$ cannot be solved in time $f(k) \cdot n^{o(k/\log k)}$, for any computable function $f$,
	even on instances $(H,G)$ where $H$ has maximum degree at most $C^{*}$. The same applies to the variants 
	$\pPart[\#]$ and $\pPart[\parity]$, assuming $\ETH[\#]$ and $\ETH[\parity]$
	respectively.
\end{thm}
Using Lemma~\ref{lem: PartitionedSub hard} and Theorem~\ref{thm: PartitionedSub Marx},
we can then prove similar lower bounds for grid
tilings.
\begin{thm}
	\label{thm: GridTiling hardness}The three variants of $\pGrid$ are
	complete for $\Wone$, $\Wone[\#]$ and $\Wone[\parity]$, respectively.
	Furthermore, the problems admit no $n^{o(k/\log k)}$ time algorithms, even on instances
	with $|\mathcal{C}|=\O(k)$, unless $\ETH$, $\ETH[\#]$ or $\ETH[\parity]$
	fails, respectively.\end{thm}
\begin{proof}
	Let $G$ and $H$ be $[k]$-vertex-colored, where we assume $V(G)=[n]$
	and $V(H)=[k]$. Replace each edge $uv$ in $G$ by the directed edges
	$uv$ and $vu$, then add all self-loops to $G$ to obtain a colored
	directed graph $G'$. Define the colorful directed graph $H'$
	by applying the same operations on $H$. Then we can observe that the color-preserving $H$-copies in $G$ stand in bijection with the color-preserving $H'$-copies in $G'$.
	
	For $i,j\in[k]$, write $E_{i,j}=E_{i,j}(G')$ for the set of directed
	edges in $G'$ from $i$-colored vertices to $j$-colored vertices.
	By Remark~\ref{rem: color-preserve-preprocessing}, we may assume
	that $E_{i,j}=\emptyset$ if $ij\notin E(H')$. Note that $E_{i,j}\subseteq[n]^{2}$;
	we use this to define an instance $(n,k,\C,\T)$ for $\pGrid$ by
	declaring $\C:=E(H')$ and $\mathcal{T}(i,j):=E_{i,j}$ for all $ij\in E(H')$.
	We then claim that the grid tilings of this instance correspond bijectively
	to the $H'$-copies in $G'$. This gives a parsimonious reduction
	from $\#\pPart$ to $\#\pGrid$, which, together with Lemma~\ref{lem: PartitionedSub hard}
	and Theorem~\ref{thm: PartitionedSub Marx}, implies all claims of
	the theorem.
	
	It remains to verify the claimed bijection: The third property
	of Definition~\ref{def: GridTiling} implies that every tiling $a:[k]^{2}\to[n]^{2}$
	encodes an edge-subset $S_{a}\subseteq E(G')$ with $|S_{a}|=|E(H)|$
	that picks exactly one element from $E_{i,j}$ for each $ij\in E(H')$.
	If the edges in $S_{a}$ are incident with exactly $k$ distinct vertices,
	then $S_{a}$ induces a $H'$-copy in $G'$. By the first two properties
	of Definition~\ref{def: GridTiling}, the edge set $S_{a}$ contains
	exactly $k$ distinct endpoints and $k$ distinct starting points.
	Since $E_{i,i}$ for $i\in[k]$ contains only self-loops, the sets
	of endpoints and starting points of edges in $S_{a}$ are identical,
	which implies that $S_{a}$ is a $H'$-copy in $G'$. Conversely,
	every $H'$-copy in $G'$ can be mapped to such a grid tiling by reversing
	this operation.
\end{proof}
In the following, we add a small technical extension to Theorem~\ref{thm: GridTiling hardness}
that allows us to assume each input instance to be balanced along
rows or columns in a certain way. While it is almost trivial to ensure this balance property by adding dummy elements, it turns out to be very useful in
our reductions from $\pGrid$.
\begin{lem}
	\label{lem: GridTiling balance}Let $\mathcal{A}=(n,k,\C,\T)$ be
	an instance for $\pGrid$ and let $\mathfrak{W}$ be either of the
	words ``horizontal'' or ``vertical''. In polynomial time, we can
	then compute a number $T\in\N$ and a grid tiling instance $\A'=(n',k,\C,\T')$
	with $n'=\O(k^{2}n)$ such that:
	\begin{enumerate}
		\item The instances $\A$ and $\A'$ have precisely the same grid tilings.
		\item For $u\in[n]$, write $(u,\star):=\{(u,v)\mid v\in[n]\}$. For $v\in[n]$,
		write $(\star,v):=\{(u,v)\mid u\in[n]\}$.
		
		\begin{enumerate}
			\item If $\mathfrak{\mathfrak{W}}$ is ``horizontal'', then for all $\kappa\in\C$
			and $u\in[n']$, we have $|\T'(\kappa)\cap(u,\star)|=T$.
			\item If $\mathfrak{\mathfrak{W}}$ is ``vertical'', then for all $\kappa\in\C$
			and $v\in[n']$, we have $|\T'(\kappa)\cap(\star,v)|=T$.
		\end{enumerate}
	\end{enumerate}
\end{lem}
\begin{proof}
	We show the statement if $\mathfrak{\mathfrak{W}}$ is ``vertical'';
	the horizontal case is shown in exactly the same manner. Let us first
	define 
	\[
	T_{\kappa,v}:=|\mathcal{T}(\kappa)\cap(\star,v)|\quad\mbox{for }\kappa\in[k]^{2}\mbox{ and }v\in[n],
	\]
	that is, the number of elements in the $v$-th column of $\T(\kappa)$.
	Then we define 
	\[
	T:=\max_{\kappa\in[k]^{2},\,v\in[n]}T_{\kappa,v}
	\]
	and let $n':=n+k^{2}T$. Consider $[n']$ to be partitioned into $[n]$
	and $k^{2}$ consecutive ``dummy'' blocks $B_{\kappa}$ for $\kappa\in[k]^{2}$,
	with $|B_{\kappa}|=T$. We keep $\C$ unchanged and modify $\mathcal{T}$
	to a function $\mathcal{T}'$ that maps from $\C$ into the power-set
	of $[n']^{2}$: For $\kappa\in[k]^{2}$ and $v\in[n]$, we simply add $T-T_{\kappa,v}$
	arbitrary distinct dummy elements from $\{(f,v)\mid f\in B_{\kappa}\}$
	to $\mathcal{T}(\kappa)$ in order to obtain $\mathcal{T}'(\kappa)$.
	
	This ensures the vertical balance property defined in the statement of the lemma, and we observe that $\mathcal{T}'$
	has the same grid tilings as $\mathcal{T}$: Every grid tiling of
	$\mathcal{T}$ is also one of $\mathcal{T}'$. Furthermore, dummy
	elements cannot be chosen in any grid tiling of $\mathcal{T}'$ since,
	for all $\kappa$ and $\kappa'$, the dummy elements in $\mathcal{T}'(\kappa)$
	and $\mathcal{T}'(\kappa')$ have disjoint first coordinates, which
	are also distinct from $[n]$. Thus, in particular, any assignment
	using dummy elements cannot satisfy the first condition of a grid
	tiling required in Definition~\ref{def: GridTiling}.\end{proof}

\section{\label{sec:Holants-and-linear}Holants, matchgates, linear combinations of
	signatures}

In the following, we give a introduction to what we call the \emph{Holant
	framework}, a toolbox introduced by \cite{Valiant2008,Cai.Lu2007,Cai.Lu2008}. Some of this material is abridged from \cite{Curticapean15}. 
We use Holant problems as an intermediate step for reducing problems, such as counting grid tilings, to the permanent.
\global\long\def\assignment{x\in\{0,1\}^{E(\Omega)}}

\subsection{Signature graphs and Holants}

The input to a Holant problem is a so-called signature graph, that is, a graph with certain functions associated with its vertices.

\begin{defn}
	A \emph{signature graph} is an edge-weighted graph $\Omega$ which
	may feature parallel edges, and which has a \emph{vertex function}
	$f_{v}:\{0,1\}^{I(v)}\to\mathbb{C}$ associated with each $v\in V(\Omega)$.
	We also call $f_{v}$ the \emph{signature} of $v$. If $v$ has degree
	$d$ and an edge-ordering $I(v)=\{e_{1},\ldots,e_{d}\}$ is specified,
	we also consider $f_{v}:\{0,1\}^{d}\to\mathbb{C}$.
\end{defn}
The \emph{Holant} of $\Omega$ is a particular sum over edge assignments
$\assignment$. For $\assignment$, we say that $e\in E(\Omega)$
is \emph{active in }$x$ if $x(e)=1$ holds, and 
we tacitly identify $x$ with the set of active edges in $x$. Given
a subset $S\subseteq E(\Omega)$, we write $x|_{S}$ for the restriction
of $x$ to $S$, which is the unique assignment in $\{0,1\}^{S}$
that agrees with $x$ on $S$.
\begin{defn}
	[adapted from \cite{Valiant2008}]\label{def: holant}Let $\Omega$
	be a signature graph with edge weights $w:E(\Omega)\to\mathbb{C}$
	and a vertex function $f_{v}:\{0,1\}^{I(v)}\to\mathbb{C}$ for each
	$v\in V(\Omega)$. For $\assignment$, we define
	\begin{eqnarray}
		\val_{\Omega}(x) & := & \prod_{v\in V(\Omega)}f_{v}(x|_{I(v)}),\label{eq: Holant value}\\
		w_{\Omega}(x) & := & \prod_{e\in x}w(e).\label{eq: Holant weight}
	\end{eqnarray}
	We say that $x$ \emph{satisfies} $\Omega$ if $\val_{\Omega}(x)\neq0$
	holds. Furthermore, we define 
	\begin{equation}
		\Holant(\Omega):=\sum_{\assignment}w_{\Omega}(x)\cdot\val_{\Omega}(x).\label{eq: Holant}
	\end{equation}
	
\end{defn}

A particularly useful type of vertex functions is that of \emph{Boolean functions}, whose ranges are restricted to $\{0,1\}$ rather than
$\mathbb{C}$. If all signatures appearing in a signature graph $\Omega'$
are Boolean, then $\Holant(\Omega')$ simply sums over those assignments
$x\in\{0,1\}^{E(\Omega')}$ that satisfy all constraints imposed by the
vertex functions, and each $x$ is weighted by $w_{\Omega'}(x)$.
As an example, we use Boolean functions to reformulate $\PerfMatch$
as a Holant problem. 

\begin{example}
	\label{exa:Holant PM} Given an edge-weighted graph $G$, let $f_{v}:\{0,1\}^{I(v)}\to\{0,1\}$
	for $v\in V(G)$ be the vertex function defined by 
	\begin{equation}
		f_{v}(x)=\begin{cases}
			1 & \mbox{if }\hw(x)=1,\\
			0 & \mbox{otherwise}.
		\end{cases}\label{eq:vtx-fn ex1}
	\end{equation}
	Let $\Omega$ denote the signature graph obtained from $G$ by
	associating $f_{v}$ with $v$, for all $v\in V(G)$. Then $\Holant(\Omega)$
	ranges over those assignments $\assignment$ in which each vertex
	is incident with exactly one active edge. Each such $x$ is weighted
	by $w_{\Omega}(x)=\prod_{e\in x}w(e)$. This is precisely the expression
	of $\PerfMatch(G)$.
\end{example}

\subsection{Gates and matchgates}

In some occasions, we can simulate signatures $f$ appearing in a
signature graph $\Omega$ by gadgets, i.e., signature graphs on ``basic''
signatures that realize $f$. We call such gadgets \emph{gates}, similar
to the $\mathcal{F}$-gates in \cite{Cai.Lu2008}, and we will be
particularly interested in \emph{matchgates}. These are gates $\Gamma$ that feature,
at each vertex $v\in V(\Gamma) $, 
the perfect matching signature from Example~\ref{exa:Holant PM} that maps $x \in \{0,1\}^{I(v)}$ to
\[
\sigHW{=1}(x):=[\hw(x)=1].
\]

The formal definition of gates and matchgates follows.

\begin{defn}
	\label{def: module}A \emph{gate} is a signature graph $\Gamma$ containing
	a set $D\subseteq E(\Gamma)$ of \emph{dangling edges}, all of which
	have edge-weight $1$. A dangling edge is an ``edge'' that
	is incident with only one vertex. We consider the dangling edges of $\Gamma$ to be labeled as $1,\ldots,|D|$. 
	
	Given a signature graph $\Omega$, a vertex $v\in V(\Omega)$
	of degree $|D|$, and an ordering of $I(v)$ as $I(v)=\{e_{1},\ldots,e_{|D|}\}$,
	we can \emph{insert} $\Gamma$ at $v$ by deleting $v$, placing a
	copy of $\Gamma$ into $G$, and identifying $e_{i}$ with the $i$-labeled
	dangling edge of $\Gamma$, for all $i$.
	
	For disjoint sets $A$, $B$, and for $x\in\{0,1\}^{A}$
	and $y\in\{0,1\}^{B}$, write $xy\in\{0,1\}^{A\cup B}$ for the assignment
	that agrees with $x$ on $A$, and with $y$ on $B$. We say that
	$xy$ \emph{extends} $x$. The \emph{signature of $\Gamma$ }is the
	function $\Sig(\Gamma):\{0,1\}^{D}\to\mathbb{C}$ that maps $x \in \{0,1\}^{D}$ to
	\begin{equation}
		\Sig(\Gamma,x)=\sum_{y\in\{0,1\}^{E(\Gamma)\setminus D}}
		w_{\Gamma}(xy)\cdot\val_{\Gamma}(xy).\label{eq: module signature}
	\end{equation}
	We also say that $\Gamma$ \emph{realizes} $\Sig(\Gamma)$. If all $v\in V(\Gamma)$
	feature the function $\sigHW{=1}$ defined above, then $\Gamma$ is a \emph{matchgate}.
	Finally, we call $\Gamma$ planar if it can be drawn in the plane
	with all dangling edges on the outer face, such that they appear in
	the order $1,\ldots,|D|$ in a clockwise traversal of this face.
\end{defn}
By the following lemma, if $\Gamma$ realizes a signature $f$, and
$v$ is a vertex with signature $f$ in a signature graph $\Omega$, then
we can insert $\Gamma$ at $v$ in a way that preserves Holants. 
In other words, we can treat $\Gamma$ as if it were a single vertex of signature $\Sig(\Gamma)$.
This will be used to reduce $\Holant(\Omega)$ to $\PerfMatch$
if all signatures in $\Omega$ can be realized by matchgates. For a proof, see Chapter 2 of \cite{Curticapean15}.
\begin{lem}
	\label{lem: holant insertions}
	Let $\Omega$ be a signature graph, let $v\in V(\Omega)$ be arbitrary,
	and let $f_{v}$ denote the vertex function of $v$ in $\Omega$.
	Furthermore, let $\Gamma$ be a (match-)gate with $\Sig(\Gamma)=f_{v}$,
	and let $\Omega'$ be obtained from $\Omega$ by inserting $\Gamma$
	at $v$. Then we have 
	\[\Holant(\Omega)=\Holant(\Omega').\] 
	If $\Omega$ and $\Gamma$ are planar and $\Omega$ is given together with a plane embedding, 
	then the following holds: If
	we order $I(v)$ according to its clockwise ordering in the embedding
	and insert $\Gamma$ under this order, then $\Omega'$ is planar.
\end{lem}
In the remainder of this subsection, we consider specific matchgates that will be relevant later.
To simplify our presentation, we abbreviate the following $4$-bitstrings. Each corresponds to a specific assignment to the edges incident with a vertex of degree $4$.
\[
\begin{array}{cccc}
\symEmpty:=0000,\quad & \symWE:=0101,\quad & \symNS:=1010,\quad & \symNSWE:=1111,\\
\symN:=1000,\quad & \symS:=0010,\quad & \symNWE:=1101,\quad & \symSWE:=0111.
\end{array}
\]

In Figure~\ref{fig: apex-matchgates-analysis}, we define a signature
$\pass$ of arity $4$ and two signatures $\pre$ and $\act$ of arity
$6$. Note that $\pass$ essentially acts as a ``crossing'' signature:
It enforces equality on its western and eastern dangling edges (numbered 4 and 2), as well as on
its northern and southern dangling edges (numbered 1 and 3). However, if all dangling edges are active, then the output of $\pass$ is $-1$ rather than $1$. This flipped sign allows $\pass$
to admit a planar matchgate $\Gamma_{\pass}$, shown in Figure~\ref{fig: apex-matchgates-analysis}.
We verified that $\Sig(\Gamma_{\pass}) = \pass$ holds by means of a computer program: For all $x\in \{0,1\}^4$, we showed mechanically that $\Sig(\Gamma_{\pass},x) = \pass(x)$ holds. Note that this verification can also be carried out by hand. For more details, consider Appendix~C of \cite{Curticapean15}.
It should also be noted
that planar matchgates for $\pass$ were already studied in \cite{Valiant2008,Cai.Gorenstein2013}.

\begin{figure}
	\begin{centering}
		\scalebox{0.81}{
			\begin{tabular}[t]{ccccc}
				\includegraphics[width=5cm]{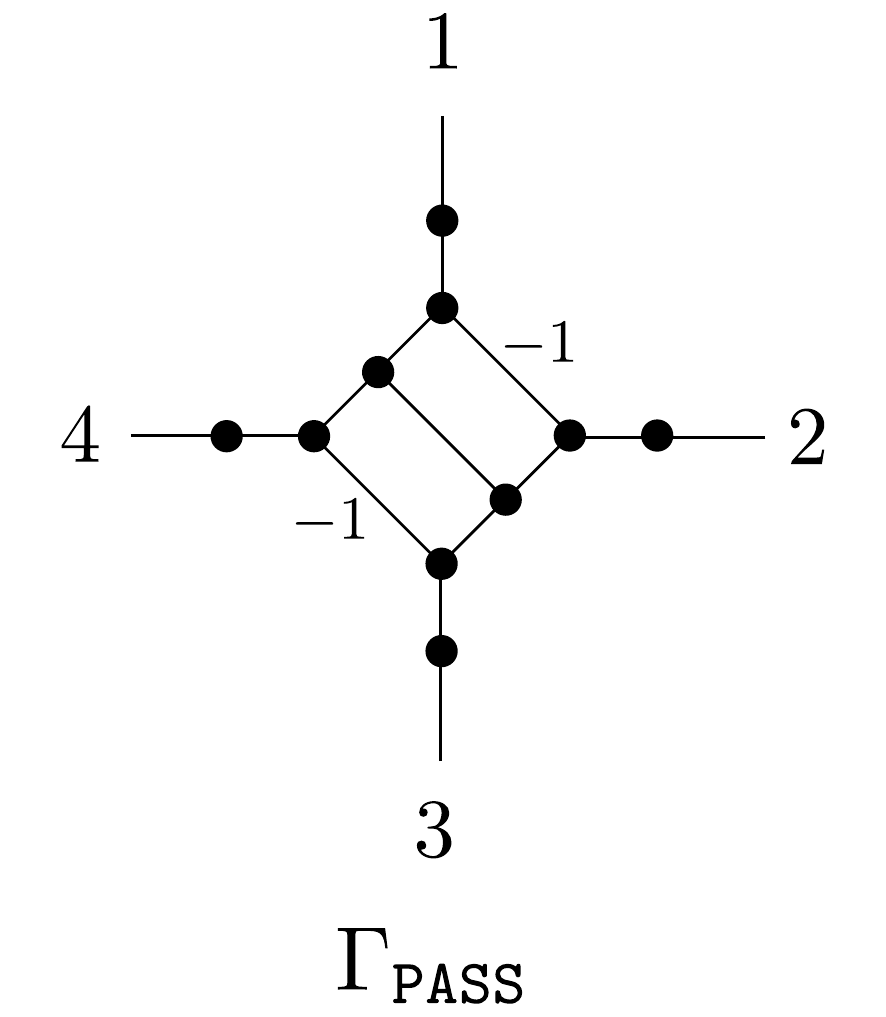} & $\quad$ & \includegraphics[width=5cm]{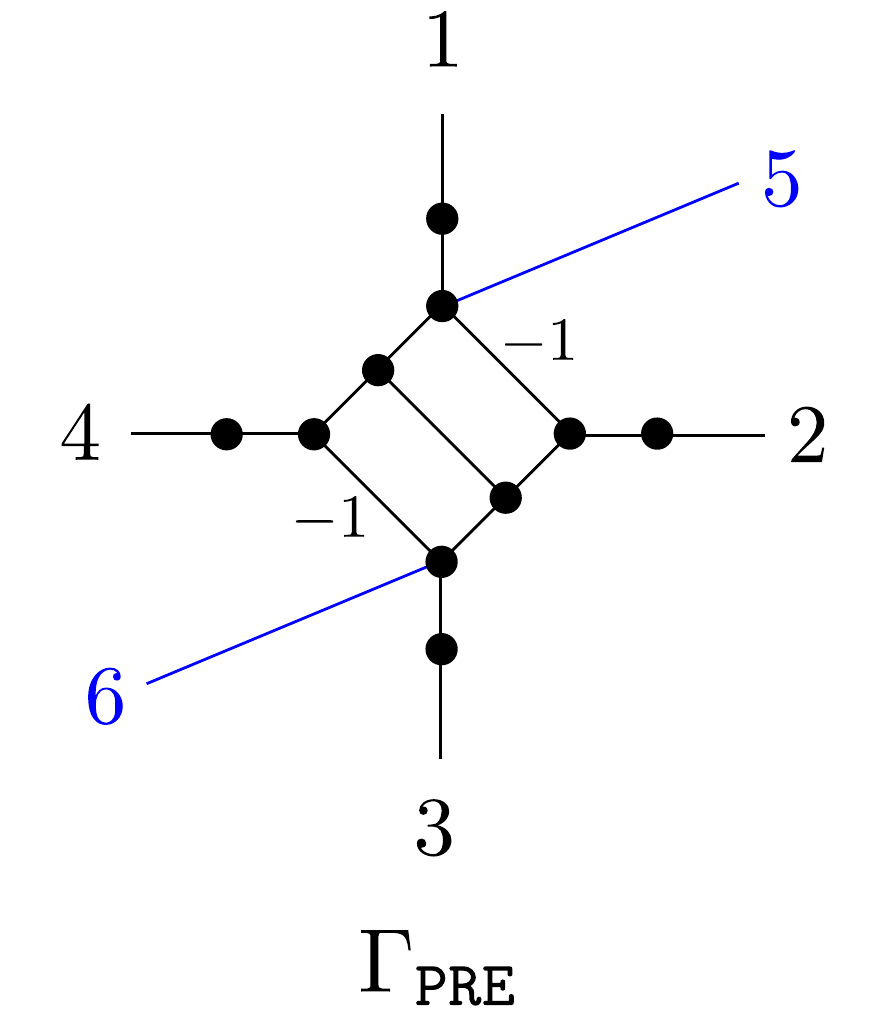} & $\quad$ & \includegraphics[width=5cm]{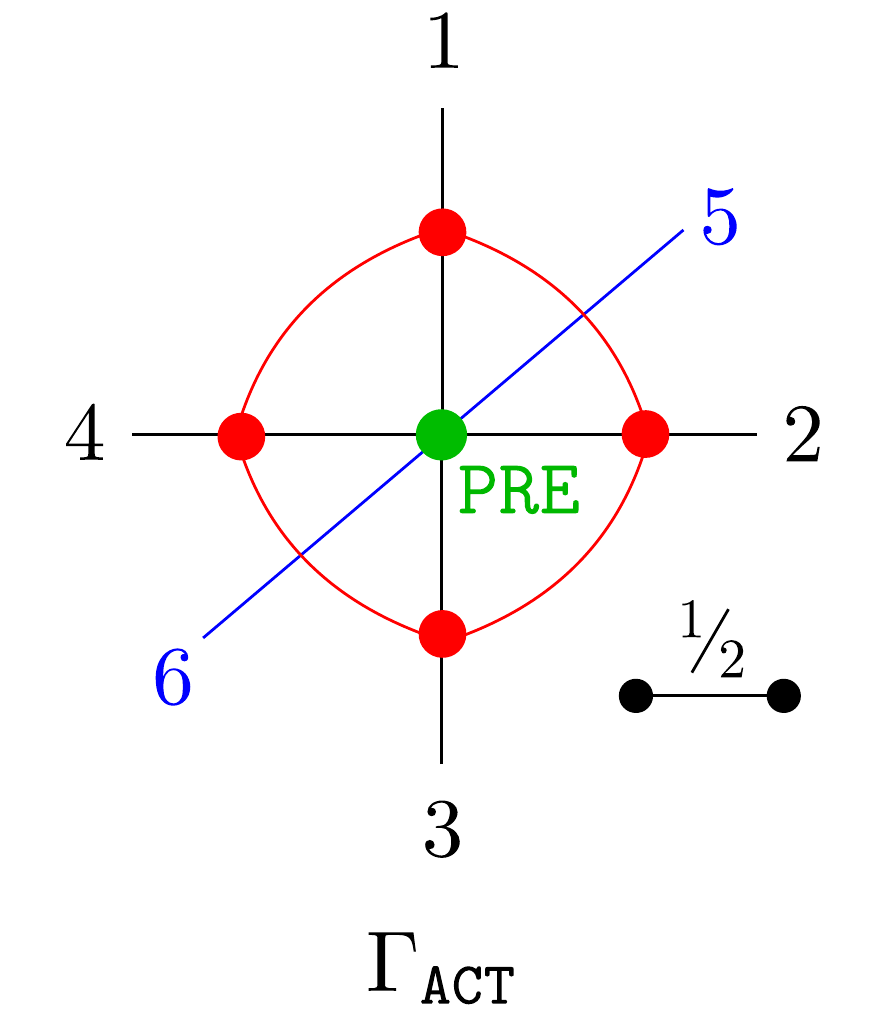}\tabularnewline
				$\quad$ &  &  &  & \tabularnewline
				$\pass(x):=\begin{cases}
				-1 &  x=\symNSWE\\
				1 &  x\in\{\symEmpty,\symWE,\symNS\}\\
				0 & \mbox{otherwise}
				\end{cases}$ &  & $\pre(x):=\begin{cases}
				\pass(y) &  x=y00\\
				1 &  x\in\{\symNS11,\symNSWE\,11\}\\
				1 &  x\in\{\symN01,\symNWE\,01\}\\
				1 &  x\in\{\symS10,\symSWE\,10\}\\
				0 & \mbox{otherwise}
				\end{cases}$ &  & $\act(x):=\begin{cases}
				\pass(y) &  x=y00\\
				1 &  x\in\{\symNS11,\symNSWE\,11\}\\
				0 & \mbox{otherwise}
				\end{cases}$\tabularnewline
			\end{tabular}
		}
		\par\end{centering}
	
	\protect\caption{\label{fig: apex-matchgates-analysis}
		The matchgates 
		$\Gamma_{\protect\pass}$, $\Gamma_{\protect\pre}$ and $\Gamma_{\protect\act}$ 
		and the signatures $\pass$, $\pre$ and $\act$. 
		Note that $\Gamma_{\protect\pass}$ has four dangling edges, numbered $1$ to $4$, 
		whereas $\Gamma_{\protect\pre}$ and $\Gamma_{\protect\act}$ each have six dangling edges, numbered $1$ to $6$.
		The signature $\pass$ is defined on assignments $x\in \{0,1\}^4$, while $\pre$ and $\act$ are defined on assignments $x\in \{0,1\}^6$. These strings correspond canonically to assignments at the dangling edges of $\Gamma_{\protect\pass}$, $\Gamma_{\protect\pre}$ and $\Gamma_{\protect\act}$.
		All black vertices
		are assigned $\protect\sigHW{=1}$. In the gate $\Gamma_{\protect\act}$,
		all red vertices are assigned $\protect\pass$, and the green middle
		vertex is assigned $\protect\pre$. Note that we can also view $\Gamma_{\protect\act}$
		as a matchgate by realizing its signatures with the matchgates $\Gamma_{\protect\pass}$
		and $\Gamma_{\protect\act}$. All matchgates are planar after removal
		of the dangling edges $5$ and $6$, which will later connect to apex
		vertices.}
\end{figure}

Next, we consider the signatures $\pre$ and $\act$, each of arity
$6$. We consider their last two inputs (the dangling edges with numbers $5$ and $6$) as ``switches'', which will later be connected to apices. It is crucial to observe that 
\[
\pre(x00) = \act(x00)=\pass(x)\quad\forall x\in\{0,1\}^{4}.
\]

That is, if the two switch edges are not active, then $\pre$ and $\act$
behave exactly like $\pass$ on their non-switch inputs. If both switches
are active, then some differences occur, namely, the restriction to non-switch
edges must be in state $\symNS$ or $\symNSWE$ for $\pre$ or $\act$
to yield a nonzero value. Furthermore, if only one of the two switches
is active, then $\act$ yields value zero, while $\pre$ still allows such assignments (such as $\symNWE 01$). We verified with a computer program that $\pre=\Sig(\Gamma_{\pre})$ holds
for the matchgate $\Gamma_{\pre}$ from Figure~\ref{fig: apex-matchgates-analysis}.
In the following, we prove manually that $\act=\Sig(\Gamma_{\act})$ holds.

\begin{lem}
	We have $\act=\Sig(\Gamma_{\act})$ with the matchgate $\Gamma_{\act}$
	from Figure~\ref{fig: apex-matchgates-analysis}.
\end{lem}
\begin{proof}
	Note that $\Gamma_{\act}$ has a green vertex of signature $\pre$,
	and some additional part (a ring of $\pass$ signatures, and an edge
	of weight $\frac{1}{2}$) which we call the \emph{even filter}. Observe
	also that, for all $x\in\{0,1\}^{4}$ and $y\in\{0,1\}^{2}$, we have the identity
	\begin{equation}
		\pre(xy)=\begin{cases}
			\act(xy) & \mbox{if }\hw(x)\mbox{ even},\\
			\mbox{arbitrary} & \mbox{otherwise}.
		\end{cases}\label{eq: sig_psi}
	\end{equation}
	The even filter now ensures the following, for all $x\in\{0,1\}^{4}$
	and $y\in\{0,1\}^{2}$:
	\begin{itemize}
		\item If $\hw(x)$ is not even, then $\Sig(\Gamma_{\act},xy)=0$, regardless
		of the value of $\pre$ on $xy$. 
		\item If $\hw(x)$ is even, then $\Sig(\Gamma_{\act},xy)=\pre(xy)$. By (\ref{eq: sig_psi}), this implies $\Sig(\Gamma_{\act},xy)=\act(xy)$.
	\end{itemize}
	Since $\act(xy)\ne0$ implies $x\in\{\symEmpty,\symNS,\symWE,\symNSWE\}$,
	which in turn implies that $\hw(x)$ is even, this will prove the
	lemma. To compute $\Sig(\Gamma_{\act},xy)$ for $x\in\{0,1\}^{4}$
	and $y\in\{0,1\}^{2}$, we consider the satisfying assignments $w$ to $E(\Gamma_{\act})$
	that extend $xy$. The dummy edge of weight $\nicefrac{1}{2}$ is
	present in any assignment $w$ and contributes a factor $\nicefrac{1}{2}$
	to $\val(w)$. (In this proof, we write $\val(w)$ instead of $\val_{\Gamma_{\act}}(w)$
	to avoid double indexing.) At each red vertex, the signature $\pass$ ensures that opposing
	edges have the same assignment under $w$. This fixes the value of
	all black edges and ensures that $\val(w)$ contains the factor $\pre(xy)$,
	contributed from the green vertex with signature $\pre$.
	
	It remains to assign values to the red edges: Due to the signature
	$\pass$ at red vertices, this is possible with at most two satisfying
	assignments $w_{1},w_{2}\in\{0,1\}^{E(\Gamma_{\act})}$:
	\begin{description}
		\item [{$w_{1}:$}] All red edges are active. Then every red vertex in
		state $\symNSWE$ yields a factor $\pass(\symNSWE)=-1$, while all
		other red vertices are in one of the states $\symNS$ or $\symWE$
		and yield value $1$. The number of red vertices in state $\symNSWE$
		is $\hw(x)$, so the value of $\Gamma_{\act}$ on $w_{1}$ is 
		\[
		\val(w_{1})=\frac{1}{2}\cdot(-1)^{\hw(x)}\cdot\pre(xy).
		\]
		
		\item [{$w_{2}:$}] No red edges are active. Then every red vertex is in
		one of the states $\symNS$ or $\symWE$ and hence yields value $1$.
		Thus, the value of $\Gamma_{\act}$ on $w_{2}$ is 
		\[
		\val(w_{2})=\frac{1}{2}\cdot\pre(xy).
		\]
		
	\end{description}
	It follows that for all $x\in\{0,1\}^{4}$ and $y\in\{0,1\}$, we
	have 
	\begin{eqnarray*}
		\Sig(\Gamma_{\act},xy) & = & \val(w_{1})+\val(w_{2})\\
		& = & \frac{1}{2}\cdot\left((-1)^{\hw(x)}\cdot\pre(xy)+\pre(xy)\right)\\
		& = & \begin{cases}
			\pre(xy) & \mbox{if }\hw(x)\mbox{ even,}\\
			0 & \mbox{otherwise}.
		\end{cases}\\
		& = & \act(xy)
	\end{eqnarray*}
	This proves the lemma.
\end{proof}

\subsection{Linear combinations of matchgate signatures}

We introduce our main tool for the later sections, a technique that allows us to simulate signatures by linear combinations of other signatures, in particular, of matchgate signatures.
\begin{defn}
	Let $f=c_{1}\cdot f_{1}+\ldots+c_{t}\cdot f_{t}$ be a signature,
	where $c_{1},\ldots,c_{t}\in\mathbb{C}$ are coefficients and $f_{1},\ldots,f_{t}$
	are signatures, and the linear combination is point-wise.
	Then we say that $f$ is $t$-combined from constituents $f_{1},\ldots,f_{t}$.
\end{defn}
We apply such linear combinations as follows: Assume we are given
a signature graph that features $k$ occurrences of some interesting
signature $f$ which cannot be realized by matchgates. If we can express
$f$ as a linear combination of $t$ constituents that do admit matchgates,
then the following lemma allows us to compute $\Holant(\Omega)$ from
the Holants of $t^{k}$ derived signature graphs whose signatures
all admit matchgates.
\begin{lem}
	\label{lem: combined signature lemma}Let $\Omega$ be a signature
	graph, let $k,t\in\mathbb{N}$ and let $w_{1},\ldots,w_{k}$ be distinct
	vertices of $\Omega$ such that the following holds: For all $\kappa\in[k]$,
	the signature $f_{\kappa}$ at $w_{\kappa}$ admits coefficients $c_{\kappa,1},\ldots,c_{\kappa,t}\in\mathbb{C}$
	and signatures $g_{\kappa,1},\ldots,g_{\kappa,t}$ such that $f_{\kappa}=\sum_{i=1}^{t}c_{\kappa,i}\cdot g_{\kappa,i}$. 
	Given a tuple $\theta\in[t]^{k}$, let $\Omega_{\theta}$ be defined
	by replacing, for each $\kappa\in[k]$, the vertex function $f_{\kappa}$
	at $w_{\kappa}$ with $g_{\kappa,\theta(\kappa)}$. Then we have 
	\begin{equation}
		\Holant(\Omega)=\sum_{\theta\in[t]^{k}}\left(\prod_{\kappa=1}^{k}c_{\kappa,\theta(\kappa)}\right)\cdot\Holant(\Omega_{\theta}).\label{eq: combined sig lin-comb}
	\end{equation}
\end{lem}

\begin{proof}
	Choose any fixed single $\kappa\in[k]$. For $i\in[t]$, let $\Omega_{i}$
	denote the signature graph obtained from $\Omega$ by replacing $f_{\kappa}$
	with $g_{\kappa,i}$. By elementary manipulations, we have
	\begin{eqnarray*}
		\Holant(\Omega) & = & \sum_{x\in\{0,1\}^{E(\Omega)}}f_{\kappa}(x)\cdot\prod_{v\in V(\Omega)\setminus\{w\}}f_{v}(x)\\
		& = & \sum_{x\in\{0,1\}^{E(\Omega)}}\left(\sum_{i=1}^{t}c_{\kappa,i}\cdot g_{\kappa,i}(x)\right)\cdot\prod_{v\in V(\Omega)\setminus\{w\}}f_{v}(x)\\
		& = & \sum_{i=1}^{t}c_{\kappa,i}\cdot\sum_{x\in\{0,1\}^{E(\Omega)}}g_{\kappa,i}(x)\prod_{v\in V(\Omega)\setminus\{w\}}f_{v}(x)\\
		& = & \sum_{i=1}^{t}c_{\kappa,i}\cdot\Holant(\Omega_{i}).
	\end{eqnarray*}

	Then apply this identity inductively for $\kappa = 1, \ldots, k$. Each
	step reduces the number of combined signatures by one, and elementary
	algebraic manipulations imply (\ref{eq: combined sig lin-comb}).
\end{proof}

When using Lemma~\ref{lem: combined signature lemma} for positive
results, as in Section~\ref{sec: genus}, then the right-hand side
of (\ref{eq: combined sig lin-comb}) is ``easy'', in the sense
that the values $\Holant(\Omega_{\theta})$ for all $\theta$ can
be obtained efficiently, e.g., by reduction to planar $\PerfMatch$.
In the same way, Lemma~\ref{lem: combined signature lemma} also
allows us to prove hardness results under Turing reductions, as we
do in Sections~\ref{sec: permanent k-apex} and \ref{sec: permanent modulo}:
In this case, the left-hand side is ``hard'' and could be computed
from oracle access to the values $\Holant(\Omega_{\theta})$ for all
$\theta$.

\section{\label{sec: genus}PerfMatch on bounded-genus graphs}

In this section, we present a first application of the framework of combined signatures:
We show that, for graphs of genus $k$, the quantity $\PerfMatch(G)$
can be expressed as a linear combination of $4^k$ values $\PerfMatch(G_i)$,
where $G_i$ is a planar graph for all $i\in [4^k]$.
The linear combinations resemble those used in \cite{Galluccio.Loebl,DBLP:journals/jct/Tesler00,ReggeZecchina},
but unlike these papers, we can state our linear combinations without any necessity for Pfaffian orientations. That is, we obtain a parameterized reduction with black-box access to counting perfect matchings in planar graphs.

\subsection{The algorithm}
Following \cite{DBLP:journals/jct/Tesler00}, we assume that the graph $G$ in question is given to us together with a plane model:
All vertices of $G$ are drawn in a polygon $P$ with $2k$ sides.
If there is a set of $d_{i}$ parallel edges $x_{i}=x_{i1}x_{i2}\cdots x_{id_{i}}$
leaving $P$ from one side and going into $P$ through another side,
we denote the two sides by $a_{i}$ and $a_{i}^{-1}$ respectively.
Since the edges are parallel, when we walk along the sides of $P$
counterclockwise, we meet the exits of edges in the order $x{}_{i1}x_{i2}\cdots x_{id_{i}}$
on side $a_{i}$, then the entrances of edges in the order $x_{id_{i}}x_{i(d_{i}-1)}\cdots x_{i1}$
on side $a_{i}^{-1}$.
If $G$ can be embedded on an orientable compact
boundaryless surface $S$ of genus $k$, then it can be drawn such
that there are no edges crossing inside $P$, and the sides of $P$
are
\[a_{1}a_{2}a_{1}^{-1}a_{2}^{-1}a_{3}a_{4}a_{3}^{-1}a_{4}^{-1}\cdots a_{2k-1}a_{2k}a_{2k-1}^{-1}a_{2k}^{-1}.\]
The side pair $a_{i},a_{i}^{-1}$ represents boundaries to be glued
together. When $G$ is drawn on the surface $S$, the edge bunches
$x_{1}$ and $x_{2}$ overpass each other without any edges crossing;
see the left picture of Figure~\ref{pic: grid cap} for such a situation, which we call a \emph{grid cap}.

We use linear combinations of matchgates to simulate the grid cap by a planar graph.
Write $x_{i}^{-1}$ to denote $x_{id_{i}}x_{i(d_{i}-1)}\cdots x_{i1}$.
Then the grid cap realizes a function that is defined on assignments $(x_{1},x{}_{2},y_{1},y_{2})$ to its dangling edges as follows:
\[
O(x_{1},x{}_{2},y_{1},y_{2})=[y_{1}=x_{1}^{-1}] \cdot [y_{2}=x_{2}^{-1}].
\]
The straightforward idea is to place a $\pass$ matchgate at each crossing of overpassing edges,
as shown in the middle of Figure~\ref{pic: grid cap}. Let us denote by
$C(x_{1},x{}_{2},y_{1},y_{2})$ the signature of the resulting gate.
In any satisfying assignment $(x_{1},x{}_{2},y_{1},y_{2})$ to its dangling edges, there
are $\hw(x_{1}) \cdot \hw(x_{2})$ instances of $\pass$ in state $\symNSWE$,
each of which gives a factor $-1$, while all other instances of $\pass$ (in states $\symNS$, $\symWE$, $\symEmpty$) give a factor $1$, so
\[
C(x_{1},x{}_{2},y_{1},y_{2})
=
(-1)^{\odd(x_{1}) \cdot \odd(x_{2})} \cdot [y_{1}=x_{1}^{-1}] \cdot [y_{2}=x_{2}^{-1}].
\]

We can thereforce conclude that $O$ can be expressed as a linear combination of signatures of type $C$, each of which is the signature of a planar matchgate.

\begin{figure}[t]
	\begin{centering}
		\includegraphics[width=0.95\textwidth]{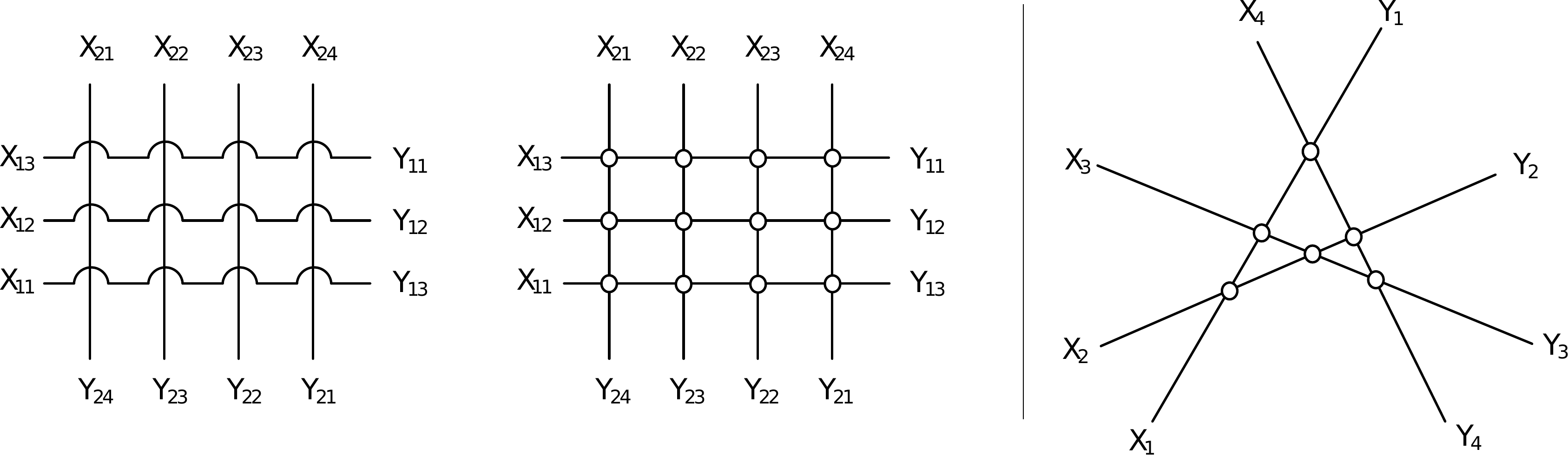}
		\par\end{centering}
	
	\protect\caption{\label{pic: grid cap} \label{pic: cross cap}The first two subfigures show a grid cap and the matchgate realizing one of the constituents used to realize the grid cap. The third subfigure shows the matchgate used to simulate a cross cap. In these matchgates, all vertices are assigned the signature $\pass$.}
\end{figure}

\begin{lem}
	\label{lem:grid-cap}Every grid cap gate is a linear combination of
	$4$ matchgates, given by
	\[
	O(x_{1},x{}_{2},y_{1},y_{2})=\frac{1}{2}(1+(-1)^{\odd(x_{1})}+(-1)^{\odd(x_{2})}+(-1)^{\odd(x_{1})+\odd(x_{2})+1}) \cdot C(x_{1},x{}_{2},y_{1},y_{2}).
	\]
\end{lem}
\begin{proof}
	Observe first that
	\[
	O(x_{1},x{}_{2},y_{1},y_{2})=\frac{1}{2}(1+(-1)^{\odd(x_{1})}+(-1)^{\odd(x_{2})}+(-1)^{\odd(x_{1})+\odd(x_{2})+1}) \cdot (-1)^{\odd(x_{1}) \cdot \odd(x_{2})} .
	\]
	From this, we can conclude that
	\begin{eqnarray*}
		O(x_{1},x{}_{2},y_{1},y_{2}) & = & \frac{1}{2}C(x_{1},x{}_{2},y_{1},y_{2})+\frac{1}{2}(-1)^{\odd(x_{1})}C(x_{1},x{}_{2},y_{1},y_{2})+\\
		& + & \frac{1}{2}(-1)^{\odd(x_{2})}C(x_{1},x{}_{2},y_{1},y_{2})-\frac{1}{2}(-1)^{\odd(x_{1})}(-1)^{\odd(x_{2})}C(x_{1},x{}_{2},y_{1},y_{2}).
	\end{eqnarray*}
	The extra factor $(-1)^{\odd(x_{1})}$ can be realized by giving weight
	$-1$ instead of $1$ to each edge $x_{1i}$ in the matchgate $C$. Hence,
	all the four functions can be realized by some matchgates similar
	to $C$ after introduction of additional $-1$ weights at some edges.
\end{proof}

We now consider non-orientable surfaces and their plane models:
If $G$ can be embedded on a non-orientable surface $S$, which is
the connected sum of a surface of orientable genus $k$ with either a projective plane
or a Klein bottle, then it can be drawn without
crossings inside $P$, such that the sides of $P$ are
\[
a_{1}a_{2}a_{1}^{-1}a_{2}^{-1}a_{3}a_{4}a_{3}^{-1}a_{4}^{-1}\cdots a_{2k-1}a_{2k}a_{2k-1}^{-1}a_{2k}^{-1}a_{2k+1}a_{2k+2} \mbox{,\ \ and}
\]
\[
a_{1}a_{2}a_{1}^{-1}a_{2}^{-1}a_{3}a_{4}a_{3}^{-1}a_{4}^{-1}\cdots a_{2k-1}a_{2k}a_{2k-1}^{-1}a_{2k}^{-1}a_{2k+1}a_{2k+2}a_{2k+3}a_{2k+4},
\]
respectively. Here, the side pair $a_{i}a_{i}$ means that, when a bunch
of edges $x_{i}=x_{i1}x_{i2}\cdots x_{id_{i}}$ leaves the interior of $P$ through
the first side $a_{i}$ and then enters back into $P$ through the second
side $a_{i}$, then we meet the exits and entrances in the order $x_{i}x_{i}$.
Such a bunch of edges is called a \emph{cross cap}, and it realizes a function
\[
O(x,y)=[y=x].
\]
If we draw it on the plane and replace each crossing
by a $\pass$ matchgate, as shown in the right part of Figure~\ref{pic: cross cap},
we get a matchgate realizing
\[
C(x,y)=(-1)^{{\hw(x) \choose 2}} \cdot [y=x].
\]
From this, we obtain a linear combination for cross cap gates from planar matchgates:
\begin{lem}
	\label{lem:cross cap}Every cross cap gate is a linear combination
	of $2$ matchgates, given by
	\[
	O(x,y)=\frac{1-i}{2} \cdot i{}^{\hw(x)} \cdot C(x,y)
	+
	\frac{1+i}{2} \cdot (-i)^{\hw(x)} \cdot C(x,y).
	\]
\end{lem}
\begin{proof}
	The sequence $(-1)^{{\hw(x) \choose 2}}$ indexed by $\hw(x)$ is
	\[1,1,-1,-1,1,1,-1,-1,\ldots\] 
	It must be a linear combination of
	$4$ sequences $w^{\hw(x)}$, for $w\in \{1,i,-1,-i\}$, all of which have the same period
	$4$, since the length 4 initial segments of the 4 sequences form
	a full rank Vandermonde matrix. 
	In fact, it can be expressed as a linear combination of two such sequences, as we can observe that
	\[(-1)^{{\hw(x) \choose 2}}=\frac{1-i}{2}i{}^{\hw(x)}+\frac{1+i}{2}(-i)^{\hw(x)}.\]
	
	The extra factor $i^{\hw(x)}$ can be realized by giving weight $i$
	instead of $1$ to each input edge in $C$. \end{proof}

Using the fact that $G$ is embedded as a plane model, and using the combined signatures for grid caps and cross caps from the last two lemmas, we then obtain the following known theorem.

\begin{thm}
	\cite{DBLP:journals/jct/Tesler00} Let $G$ be a graph that is embedded on a surface.
	Then $\PerfMatch(G)$ is a summation of $\PerfMatch$ of $2^{2k}$, $2{}^{2k+1}$
	or $2{}^{2k+2}$ planar graphs, respectively,
	if the surface is the connected sum of an orientable surface of genus $k$ with
	the plane, the projective
	plane, or the Klein bottle, respectively. \end{thm}
\begin{proof}
	By Lemma~\ref{lem:grid-cap} and \ref{lem:cross cap}, use Lemma~\ref{lem: combined signature lemma}
	on the $k$ grid caps and $0$, $1$ or $2$ cross caps.\end{proof}

\subsection{Additional remarks}
For a matrix $A$, let $A^{\otimes k}$ denote the matrix obtained from the $k$-fold Kronecker product $A \otimes \ldots \otimes A$. The essence of Lemma~\ref{lem:grid-cap} is that we can use the four matchgates to realize all four columns of the basis
\[\left(\begin{array}{cc}
1 & 1\\
1 & -1
\end{array}\right)^{\otimes 2},
\]
so that we can then obtain any other function by linear combinations. The same observation also holds for a larger base
\[\left(\begin{array}{cc}
1 & 1\\
1 & -1
\end{array}\right)^{\otimes m}.
\]

We give an example: In a cross cap of $m$ edges, we may replace each
edge by a bunch of parallel edges, and call the result a \emph{grated cross cap}.
All the ${m \choose 2}$ latent crossings of the cross cap become grid caps in the grated cross cap.
\begin{fact}
	\label{fact:grated cross cap}Every grated cross cap gate over $m$
	bunches of edges, as defined above, can be expressed as a linear combination of $2^{m}$ planar matchgates.
\end{fact}
%

In fact, these $2^m$ basis matchgates are powerful enough to express (as a linear combination) 
any function that depends only upon the parities $p_1, \ldots, p_m$ of active edges in the $m$ edge bunches.
However, among these functions, we currently only know one interesting function, i.e., the grid cap. 
Even the grated cross cap seems too artificial to be related with a natural tractability result.
A similar generalization applies to Lemma~\ref{lem:cross cap}, where the functions to be expressed may also depend upon residuals of the numbers of active edges in the $m$ edge bunches, in this case however modulo $4$ rather than $2$.

\section{\label{sec: permanent k-apex}The permanent on k-apex graphs}

In this section, we prove Theorem~\ref{thm: hadw hard} by an application
of our framework of combined signatures. We use $\pGrid[\#]$ as a
reduction source, and from a high level, our approach could be compared
to, say, the reduction in \cite{Marx2012} for planar multiway cut.
Given an instance $\A$ to $\pGrid[\#]$, we proceed as follows:
\begin{enumerate}
	\item We express the solution to the instance as $\Holant(G)$ for a signature graph $G$ in
	Section~\ref{sub:Global-construction}. 
	\item We realize the signatures of $G$ in Section~\ref{sub:Realizing-cell-signatures}.
	At this point however, we require combined signatures, and this is
	where we depart from the usual reductions from $\pGrid$.
\end{enumerate}
Large parts of this section will be reused
in Section~\ref{sec: permanent modulo} with an added layer of technicalities.

\subsection{\label{sub:Global-construction}Global construction}

In the following, let $\mathcal{A}=(n,k,\C,\T)$ be a fixed instance
to $\pGrid[\#]$, as specified in Definition~\ref{def: GridTiling}.
By applying vertical balance as in Lemma~\ref{lem: GridTiling balance},
we may assume the existence of some number $T\leq n$ such that for
all $\kappa\in\C$ and all $v\in[n]$, there are exactly $T$ elements
of type $(\star,v)$ in $\T(\kappa)$. This will become relevant in
Section~\ref{sub:Realizing-cell-signatures}.

First, we reformulate $\mathcal{A}$ as the Holant of
a signature graph $G=G(\mathcal{A})$. This graph $G$ consists of
a $k\times k$ square grid of \emph{cells,} and $4k$ additional \emph{border
	vertices} adjacent to the borders of the grid, as seen in the left part of Figure~\ref{fig: apex signature graph}.
Note that $G$ is planar. We denote its vertices by $c_{\kappa}$
for $\kappa\in\Xi$, where
\[
\Xi:=[k]^{2}\,\cup\,\{\mathsf{N},\mathsf{W},\mathsf{S},\mathsf{E}\}\times[k].
\]

For $i\in[k]$, we declare $(\mathsf{N},i)$ to be vertically adjacent
to $(1,i)$, and $(\mathsf{S},i)$ to $(k,i)$. Likewise, we declare
$(\mathsf{W},i)$ to be horizontally adjacent to $(i,1)$, and $(\mathsf{E},i)$
to $(i,k)$. We refer to the neighbors of any index $\kappa\in\Xi$ or vertex $c_\kappa \in V(G)$ using cardinal
directions in the obvious way, e.g., we may speak of the northern
neighbor of a vertex. Between any pair of vertices $c_{\kappa}$ and
$c_{\kappa'}$ with adjacent indices $\kappa$ and $\kappa'$, we
place a set $E_{\kappa,\kappa'}$ of $n$ parallel edges, which we
call an \emph{edge bundle}.
\begin{figure}
	\begin{centering}
		\includegraphics[width=1\textwidth]{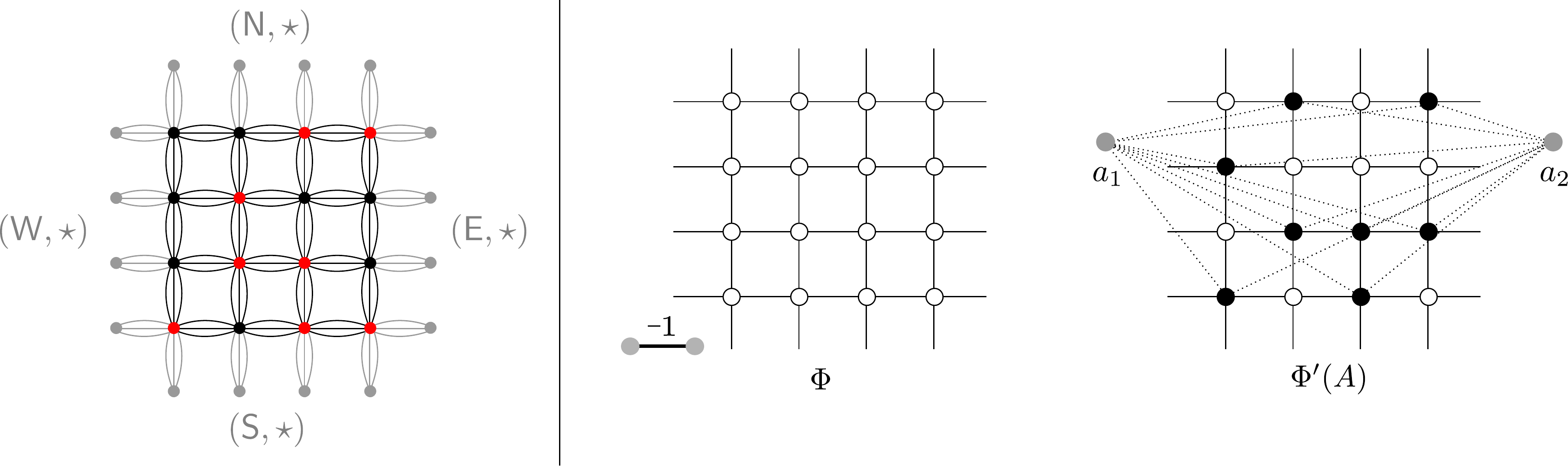}
		\par\end{centering}
	
	\protect\caption{
		\label{fig: apex signature graph}The left part of the figure shows the signature graph $G(\protect\A)$.
		Border vertices $c_{\kappa}$ for $\kappa\in\{\mathsf{N},\mathsf{W},\mathsf{S},\mathsf{E}\}\times[k]$
		and their incident edges are colored gray. Cell vertices $c_{\kappa}$
		for $\kappa\in\protect\C$ are colored red, while vertices $c_{\kappa}$
		for $\kappa\in[k]^{2}\setminus\protect\C$ are colored black. Horizontally
		or vertically adjacent vertices are connected by an edge bundle of
		$n$ parallel edges.
		\label{fig: phi-phi'}The right part of the figure shows the gates $\Phi$ and $\Phi'(A)$. Each white
		vertex is assigned $\protect\pass$, each black vertex is assigned
		$\protect\act$, and each gray vertex is assigned $\protect\sigHW{=1}$.
		Edges from apices in $\Phi'$ are drawn dashed.
		Note that, due to the balance property of $\mathcal{T}$, we may assume
		that every column has the same number $T$ of occurrences of $\protect\act$.
	}
\end{figure}

We proceed to define the signatures of $G$. In the assignments $a\in\{0,1\}^{E(G)}$
we are interested in, each edge bundle features exactly one active
edge, which is used to encode a number from $[n]$. At border vertices,
we place the signature $\sigHW{=1}$ to ensure this. The signatures
of cells $c_{\kappa}$ with $\kappa\in[k]^{2}$ are then defined so
that each cell propagates the number $x_{W}\in[n]$ encoded by its
western incident edge bundle to the east, and its northern number
$x_{N}\in[n]$ to the south, while checking along the way whether
$(x_{W},x_{N})\in\mathcal{T}(\kappa)$ holds.
\begin{rem}
	\label{rem: apex conventions}We adhere to the following notational conventions in this
	section:
	\begin{itemize}
		\item For $v\in[n]$, we often identify the string $0^{v-1}10^{n-v}\in\{0,1\}^{n}$
		with the number $v$ when it is clear from the context which of these
		two objects we currently refer to.
		\item For $\kappa\in[k]^{2}$, the $4n$ incident edges of each vertex $c_{\kappa}$
		are ordered such that all northern edges appear first, in a block
		of length $n$, followed by the $n$ eastern, the $n$ southern, and
		finally the $n$ western edges.
		\item We implicitly consider strings $x\in\{0,1\}^{4n}$ to be
		decomposed into $x=x_{N}x_{E}x_{S}x_{W}$ with four bistrings $x_{N},x_{E},x_{S},x_{W}\in\{0,1\}^{n}$ corresponding to the four cardinal directions.
	\end{itemize}
\end{rem}
Using these conventions, we then define the following predicates for
strings $x\in\{0,1\}^{4n}$:\label{text: one-prop} 
\begin{eqnarray*}
	\varphi_{\mathit{one}}(x) & \equiv & \hw(x_{N})=1\,\wedge\,\hw(x_{W})=1,\\
	\varphi_{\mathit{prop}}(x) & \equiv & x_{N}=x_{S}\wedge x_{W}=x_{E}.
\end{eqnarray*}

If a function $f$ satisfies $\varphi_{\mathit{prop}}(x)$ for each
$x\in\supp(f)$, then we call $f$ \emph{propagating}. For each $\kappa\in[k]^{2}$,
we place a specific propagating signature $f_{\kappa}$ at the vertex
$c_{\kappa}$ in order to complete $G$ to a signature graph whose
satisfying assignments correspond bijectively to the grid tilings
of $\mathcal{A}=(n,k,\C,\T)$.
\begin{defn}
	\label{def: f and g}
	Let let $\mathcal{A}=(n,k,\C,\T)$ an instance to the grid tiling problem, as described above.
	For all $\kappa\in[k]^{2}\setminus\C$, we define
	the vertex function $f_{\kappa}:\{0,1\}^{4n}\to\{0,1\}$ of $c_{\kappa}$
	such that, for all $x\in\{0,1\}^{4n}$ satisfying the predicate $\varphi_{\mathit{one}}(x)$,
	we have 
	\[
	f_{\kappa}(x) := [\varphi_{\mathit{prop}}(x)].
	\]
	Note that no requirement is imposed upon $f_{\kappa}(x)$ on those
	$x\in\{0,1\}^{4n}$ that fail to satisfy $\varphi_{\mathit{one}}(x)$.
	For all remaining $\kappa$, namely all $\kappa\in\C$, we define
	the vertex function $g_{\kappa}$ of $c_{\kappa}$ on such $x\in\{0,1\}^{4n}$
	by declaring
	\[
	g_{\kappa}(x) := [\varphi_{\mathit{prop}}(x)\wedge(x_{W},x_{N})\in\mathcal{T}(\kappa)]
	\]
	
\end{defn}
This finishes the definition of the signature graph $G=G(\A)$. In the following, we verify by a simple argument that $G$ indeed encodes $\A$ properly.
\begin{lem}
	\label{lem: Holant =00003D GridTilings}The grid tilings of $\A$
	correspond bijectively to the satisfying assignments $x\in\{0,1\}^{E(G)}$ of $G$,
	and each satisfying assignment $x$ additionally has $\val_{G}(x)=1$.\end{lem}

\begin{proof}
	Every grid tiling $a:[k]^{2}\to[n]^{2}$ can be transformed into an
	assignment $x(a)\in\{0,1\}^{E(G)}$ as follows: For each $\kappa\in[k]^{2}$,
	with $a(\kappa)=(u,v)$, declare the $u$-th edge in the western edge
	bundle of $c_{\kappa}$ and the $v$-th edge in the northern edge
	bundle of $c_{\kappa}$ to be active. At vertices $c_{(k,\star)}$,
	copy the assignment from northern edges to southern edges, and at
	$c_{(\star,k)}$, copy the assignment from western edges to eastern
	edges. Declare all other edges to be inactive. It follows from the
	definition of $f_{\kappa}$ at $\kappa\in\C$ and $g_{\kappa}$ at
	$\kappa\in[k]^{2}\setminus\C$ that $\val_{G}(x(a))=1$ holds.
	
	For the converse direction, we show that every satisfying assignment
	$x\in\{0,1\}^{E(G)}$ can be written as $x=x(a)$ for some grid tiling
	$a$, where $x(a)$ is defined as in the previous paragraph. Note
	that this also implies $\val_{G}(x)=1$. By the signature $\sigHW{=1}$,
	every border vertex is incident with exactly one active edge in $x$.
	Hence, the restriction of $x$ to $I(c_{1,1})$ satisfies $\varphi_{\mathit{one}}$;
	call this restricted assignment $y$.
	\begin{itemize}
		\item If $(1,1)\in[k]^{2}\setminus\C$, then the vertex function of $c_{1,1}$
		is $f_{1,1}$. Since $f_{1,1}(y)=1$, and since $f_{1,1}$ is propagating
		on inputs satisfying $\varphi_{\mathit{one}}$, we also have $\varphi_{\mathit{prop}}(y)$. 
		\item If $(1,1)\in\C$, then we additionally have $(y_{W},y_{N})\in\mathcal{T}(1,1)$
		by definition of $g_{1,1}$. 
	\end{itemize}
	By induction along rows and columns, we obtain, for every $\kappa\in[k]^{2}$,
	that the partial assignment $y$ at $I(c_{\kappa})$ satisfies $\varphi_{\mathit{prop}}(y)$
	and $(y_{W},y_{N})\in\mathcal{T}(\kappa)$ if $\kappa\in\C$. Hence
	$x=x(a)$ holds for a unique grid tiling $a$.
\end{proof}

In the next subsection, we realize each signature $f_{\kappa}$ for
$\kappa\in\C$ as a planar matchgate, and each $g_{\kappa}$ for $\kappa\in[k]^{2}\setminus\C$
as a linear combination of two matchgate signatures that have maximum
apex number $2$. Note that the remaining signatures $\sigHW{=1}$
occurring in $G$ are planar.
Since $G$ itself is planar and features
at most $\O(k)$ signatures $g_{\kappa}$, the graphs realizing $G$
will feature at most $\O(k)$ apices, and we will use this to obtain
the desired parameterized reduction and lower bound under $\ETH[\#]$.

\subsection{\label{sub:Realizing-cell-signatures}Realizing cell signatures}

It can be shown (under no additional assumptions) that some of the signatures $g_{\kappa}$ for $\kappa \in [k]^2$
are non-planar. From a complexity viewpoint, if all such signatures
were planar and we knew explicit planar matchgates, then we could
reduce $\pGrid[\#]$ to planar $\PerfMatch$, and thus show $\mathsf{FP}=\sharpP$
by the FKT method. Rather than trying to use planar matchgates, we show that each signature
$g_{\kappa}$ can be realized as a specific \emph{linear combination}
of the signatures of one planar and one $2$-apex matchgate. Note again that at least one non-planar constituent is necessary, as we could otherwise show $\FPT=\Wone[\#]$.

In the remainder of this section, we consider $\kappa\in[k]^{2}$
to be fixed, we write $A=\mathcal{T}(\kappa)$ and we recall that
$A\subseteq[n]^{2}$. The constituents for $g_{\kappa}$ will be the
signatures of two gates $\Phi$ and $\Phi'(A)$, which use as building
blocks the signatures $\pass$ and $\act$ from Section~\ref{sec:Holants-and-linear}.
\begin{defn}
	\label{def: cell gates}Let $n\in\N$ and let $A\subseteq[n]^{2}$.
	We define gates $\Phi$ and $\Phi'=\Phi'(A)$ with $4n$ dangling
	edges (that is, with $n$ dangling edges for each cardinal direction) as follows.
	Consider also the right part of Figure~\ref{fig: phi-phi'}.
	\begin{itemize}
		\item To obtain the gate $\Phi$, arrange vertices $b_{\tau}$ for $\tau\in[n]^{2}$
		in a $n\times n$ grid and assign the signature $\pass$ to each such
		vertex.
		Add a single edge of weight $-1$ between two fresh vertices
		of signature $\sigHW{=1}$. 
		\item A similar construction yields the gate $\Phi'$: Starting from $\Phi$, remove
		the extra edge of weight $-1$, add apex vertices $a_{1}$ and $a_{2}$
		with signatures $\sigHW{=1}$, and for all $\tau\in A$, do the following:
		
		\begin{enumerate}
			\item Replace the signature $\pass$ at $b_{\tau}$ with $\act$.
			\item Add the edges $a_{1}b_{\tau}$ and $a_{2}b_{\tau}$. Declare these
			to be the last two edges in the edge ordering of $I(v_{\tau})$. 
		\end{enumerate}
		
	\end{itemize}
\end{defn}

Recall that $\pass$ is realized by the
planar matchgate $\Gamma_{\pass}$, so we can also view the gate $\Phi$
as a planar matchgate after realizing all signatures by matchgates. 
We will later switch between these views depending
on the application. Note also that the $2$-coloring of $\Gamma_{\pass}$ can be extended
to one of $\Phi$.
Likewise, $\act$ is realized by the matchgate $\Gamma_{\act}$,
which is planar when ignoring its dangling edges $5$ and $6$. That
is, after realizing each occurrence of $\act$ by $\Gamma_{\act}$,
the resulting matchgate obtained from $\Phi'$ is planar after removal
of $a_{1}$ and $a_{2}$.

Our goal for this subsection is to realize the signatures $f_{\kappa}$ and $g_{\kappa}$
from Definition~\ref{def: f and g}. In the following, we prove that
$f_{\kappa}=\Sig(\Phi)$ and that $g_{\kappa}$ can be realized by
a linear combination of $\Sig(\Phi)$ and $\Sig(\Phi')$. It will
be crucial for our calculations to assume our instance $\A$ for $\pGrid$
to be balanced: By Lemma~\ref{lem: GridTiling balance}, we assume there is some $T\in\mathbb{N}$
such that $|A\cap(\star,v)|=T$ for all $v\in[n]$. That is, in the
right part of Figure~\ref{fig: phi-phi'}, we may assume that every
column of $\Phi'(A)$ features the same number $T$ of vertices with signature $\act$. 
\begin{lem}
	\label{lem: cell signatures}Recall the definition of the predicates
	$\varphi_{\mathit{one}}$ and $\varphi_{\mathit{prop}}$ \vpageref{text: one-prop}.
	Let $x\in\{0,1\}^{4n}$ be an assignment that satisfies the predicate
	$\varphi_{\mathit{one}}$. Then
	\begin{eqnarray}
		\Sig(\Phi,x) & = & \begin{cases}
			0 & \mbox{if }\neg\varphi_{\mathit{prop}}(x),\\
			1 & \mbox{if }\varphi_{\mathit{prop}}(x).
		\end{cases}\label{eq: First matchgate}\\
		\Sig(\Phi'(A),x) & = & \begin{cases}
			0 & \quad\mbox{if }\neg\varphi_{\mathit{prop}}(x)\\
			\begin{cases}
				-T & \mbox{if }(x_{W},x_{N})\notin A\\
				-T+2 & \mbox{if }(x_{W},x_{N})\in A
			\end{cases} & \quad\mbox{if }\varphi_{\mathit{prop}}(x).
		\end{cases}\label{eq: Second matchgate}
	\end{eqnarray}
	Note that $f_{\kappa}=\Sig(\Phi)$ for $\kappa\in[k]^{2}\setminus\C$.
	For $\kappa\in\C$ and for $x\in\{0,1\}^{4n}$ satisfying $\varphi_{\mathit{one}}$,
	we have 
	\begin{equation}
		g_{\kappa}(x)=\frac{T}{2}\cdot\Sig(\Phi,x)+\frac{1}{2}\cdot\Sig(\Phi'(\mathcal{T}(\kappa)),x).\label{eq: apex combined signature}
	\end{equation}
	
\end{lem}
In Section~\ref{sub:Computing-the-signatures}, we prove Lemma~\ref{lem: cell signatures}
by inspecting the possible satisfying assignments to $\Phi$ and $\Phi'$.
Before doing this, let us first show how Lemma~\ref{lem: cell signatures}
implies Theorem~\ref{thm: apex hard}. We will require parts of this
argument again in Section~\ref{sec: permanent modulo}.
\begin{proof}
	[Proof of Theorem \ref{thm: apex hard}]Using Lemma~\ref{lem: Holant =00003D GridTilings},
	we know that \textbf{$\Holant(G)$} counts precisely the grid tilings of $\A$. By Theorem~\ref{thm: GridTiling hardness}, this problem is $\Wone[\#]$-hard and cannot be solved in time $f(k)\cdot n^{o(k / \log k)}$, even on instances $\mathcal{A}=(n,k,\C,\T)$ with $|\C| = \mathcal{O}(k)$.
	
	Using the linear combination (\ref{eq: apex combined signature}) and
	Lemma~\ref{lem: combined signature lemma} about the linear combinations
	of signatures, as well as Lemma~\ref{lem: holant insertions} about inserting matchgates into signature graphs, we obtain 
	\begin{equation}
		\Holant(G)=\frac{1}{2^{|\C|}}\sum_{\omega:\C\to[2]}T^{d(\omega)}\cdot\PerfMatch(H_{\omega}).\label{eq: Combined GridTiling}
	\end{equation}
	For $\omega:\C\to[2]$,
	the number $d(\omega)$ is the number of $1$-entries in $\omega$,
	and the graph $H_{\omega}$ is obtained as follows: 
	\begin{itemize}
		\item For $\kappa\in[k]^{2}\setminus\C$, insert the matchgate $\Phi$ at
		the cell vertex $c_{\kappa}$.
		\item For $\kappa\in\C$ with $\omega(\kappa)=1$, insert the matchgate
		$\Phi$ at $c_{\kappa}$ as well. 
		\item For $\kappa\in\C$ with $\omega(\kappa)=2$, insert the matchgate
		$\Phi'(\mathcal{T}(\kappa))$ at $c_{\kappa}$. 
	\end{itemize}
	Since $G$ is planar, and since $\Phi$ is planar and $\Phi'(\mathcal{T}(\kappa))$
	for $\kappa\in\C$ has at most $2$ apices, it follows that
	$\apex(H_{\omega})\leq2|\C|$ for all $\omega:\C\to[2]$, and this
	proves the required parameter bound. By $2$-coloring the matchgates
	$\Phi$ and $\Phi'$, it can furthermore be verified that each graph $H_{\omega}$
	is bipartite.
	
	Additionally, by construction of the matchgates $\Gamma_{\pass}$
	and $\Gamma_{\act}$, every graph $H_{\omega}$ features only edge-weights
	from the set $\{-1,\nicefrac{1}{2},1\}$. Non-unit edge-weights in
	$H_{\omega}$ appear only at edges $uv\in E(H_{\omega})$ not incident
	with apices. We can hence use standard weight simulation techniques
	to remove the edge-weights $-1$ and $\nicefrac{1}{2}$, as in \cite{Valiant1979a} or Chapter~1 of \cite{Curticapean15}, while maintaining
	the apex number.
	We consequently obtain $\Wone[\#]$-completeness of the permanent under the apex parameter and the claimed lower bound under $\ETH[\#]$.\end{proof}
\begin{rem}
	\label{rem: PerfMatch / apex extra conditions}The following might
	prove useful for later applications: By construction, the apices in
	the constructed graphs $H_{\omega}$ form an independent set, for any
	$\omega:[k]^{2}\to[2]$, and each non-apex vertex in $H_{\omega}$
	is incident with at most one apex. This last condition holds because
	the matchgate $\Gamma_{\act}$ has no vertex with two incident dangling
	edges.
\end{rem}

\subsection{\label{sub:Computing-the-signatures}Calculating the signatures of
	$\Phi$ and $\Phi'$}

In the remainder of this section, we provide the deferred proof of
Lemma \ref{lem: cell signatures}. To this end, we calculate the signatures
of $\Phi$ and $\Phi'$ by analyzing, for any given assignment $x\in\{0,1\}^{4n}$
to their dangling edges, the possible satisfying assignments $xy$
extending $x$.

\subsubsection{\label{sub: signature phi}Calculating the signature of $\Phi$}

Let $x\in\{0,1\}^{4n}$ be an assignment to the dangling edges of
$\Phi$ that satisfies $\varphi_{\mathit{one}}(x)$, and let $xy\in\{0,1\}^{E(\Phi)}$
be a satisfying assignment to $\Phi$ that extends $x$. We show that,
whenever $\varphi_{\mathit{prop}}(x)$ holds, then $y$ is unique
and $xy$ has value $1$, so $\Sig(\Phi,x)=\val_{\Phi}(xy)=1$. Furthermore,
we show that, if $x$ does not satisfy the predicate $\varphi_{\mathit{prop}}$,
then no such $y$ exists, and hence $\Sig(\Phi,x)=0.$
\begin{figure}
	\begin{centering}
		\includegraphics[width=0.55\textwidth]{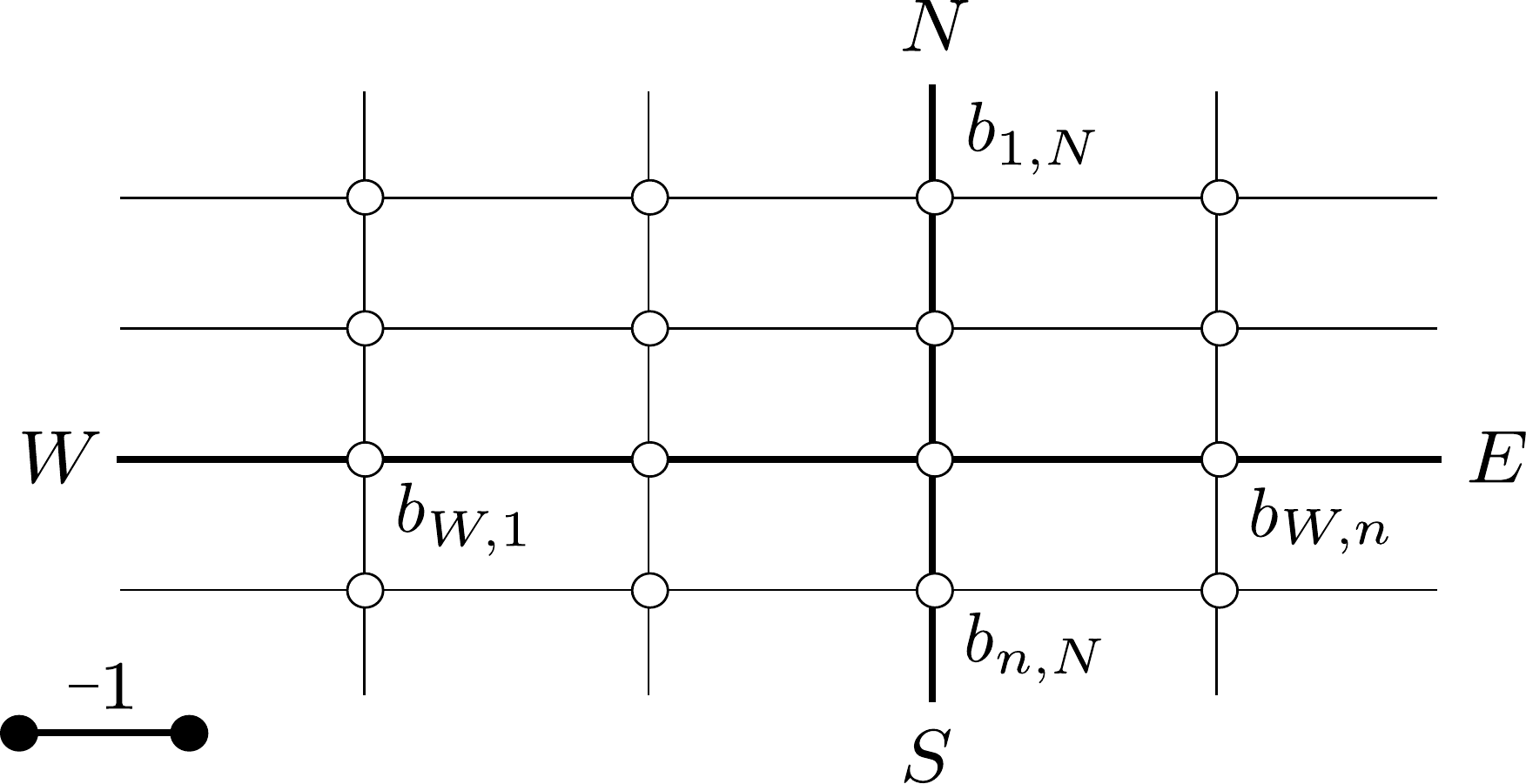}
		\par\end{centering}
	
	\caption{\label{fig: phi-assignments}The unique assignment $y$ to $E(\Phi)$
		that extends $x$. Active edges are drawn with thicker lines than
		non-active edges. Note that the edge of weight $-1$ with $\protect\sigHW{=1}$
		at its endpoints must be active in any satisfying assignment.}
\end{figure}

Recall from Remark~\ref{rem: apex conventions} that we implicitly
decompose the string $x$ into $x_{N},x_{E},x_{S},x_{W}$. Write $W\in[n]$
and $N\in[n]$ for the unique non-zero index in $x_{W}\in\{0,1\}^{n}$
and $x_{N}\in\{0,1\}^{n}$, respectively. These numbers are well-defined
because $x$ satisfies $\varphi_{one}(x)$ by assumption. Then all
western and eastern edges of vertices in row $(W,\star)$ are active
in $xy$, see Figure~\ref{fig: phi-assignments}: The western edge
of the vertex $b_{W,1}$ is active by definition, and since $xy$
satisfies $\Phi$ and $\pass$ at $b_{W,1}$, this vertex must be
in state $\symWE$ or $\symNSWE$, so its eastern edge is also active.
The same follows inductively for all vertices in the row $(W,\star)$.
By the same argument, rotated about 90 degrees, all northern and southern
edges of vertices in row $(\star,N)$ are active in $xy$.

By a similar argument, no other edges are active, and we conclude
that $y$ is uniquely determined by $x$. Furthermore, if $E$ and
$S$ denote the active indices in $x_{E}$ and $x_{S}$, then we observe
that $W=E$ and $N=S$, since otherwise $xy$ could not satisfy $b_{W,n}$
and $b_{n,N}$. Hence, $xy$ satisfies $\Phi$ only if $\varphi_{\mathit{prop}}(x)$
holds. We obtain
\[
\Sig(\Phi,x)=0\qquad\mbox{if }\neg\varphi_{\mathit{prop}}(x).
\]

If $\varphi_{\mathit{prop}}(x)$ holds, then $b_{W,N}$ is in state
$\symNSWE$ under $xy$, while the $n-1$ other vertices in row $(W,\star)$
are in state $\symWE$, the $n-1$ other vertices in column $(\star,N)$
are in state $\symNS$, and the remaining $n^{2}-2n+1$ vertices are
in state $\symEmpty$. Furthermore, we have the additional active
edge of weight $-1$. Hence, in conclusion, $\varphi_{\mathit{prop}}(x)$
implies 
\begin{eqnarray*}
	\Sig(\Phi,x) & = & \val(\Phi,xy)\\
	& = & (-1)\cdot\pass(\symNSWE)\cdot\pass(\symNS)^{n-1}\cdot\pass(\symWE)^{n-1}\cdot\pass(\symEmpty)^{n^{2}-2n+1}\\
	& = & 1.
\end{eqnarray*}
This proves (\ref{eq: First matchgate}).

\subsubsection{\label{sub: signature phi'}Calculating the signature of $\Phi'(A)$}

Let $\Phi'=\Phi'(A)$ for some fixed $A\subseteq[n]^{2}$, let $D\subseteq E(\Phi')$
denote the dangling edges of $\Phi'$ and let $F=I(a_{1})\cup I(a_{2})$
denote the set of edges incident with either of the apices $a_{1}$
or $a_{2}$ in $\Phi'$. Let \[x\in\{0,1\}^{4n}\] be an assignment
to $D$ that satisfies the predicate $\varphi_{\mathit{one}}(x)$,
and let $xyz\in\{0,1\}^{E(\Phi')}$ be a satisfying assignment to
the edges of $\Phi'$ that extends $x$, with 
\begin{eqnarray*}
	y & \in & \{0,1\}^{E(\Phi')\setminus(F\cup D)},\\
	z & \in & \{0,1\}^{F}.
\end{eqnarray*}
We consider the restriction of $xyz$ to $xy$, that is, to edges
not incident with any apex. By definition of $\pass$ and $\act$,
we have, for every vertex $b\in V(\Phi')\setminus\{a_{1},a_{2}\}$,
that 
\begin{equation}
	(xy)|_{I(b)}\in\{\symEmpty,\symNS,\symWE,\symNSWE\}.\label{eq: phi' vertices}
\end{equation}

Recall from Remark~\ref{rem: apex conventions} that we decompose
$x$ into $x_{N},x_{E},x_{S},x_{W}$, and write $W\in[n]$ and $N\in[n]$
for the unique non-zero index in $x_{W}\in\{0,1\}^{n}$ and $x_{N}\in\{0,1\}^{n}$,
respectively. Since $(xy)|_{I(b)}\in\supp(\pass)$ holds by (\ref{eq: phi' vertices})
and the definition of $\pass$, the same argument as in the previous
subsection for $\Phi$ shows that the western and eastern edges of
all vertices in row $(W,\star)$ are active under $xy$, as well as
the northern and southern edges of all vertices in the column $(\star,N)$.
Likewise, as seen in the previous subsection, it shows that no other
edges in $E(\Phi')\setminus F$ are active, that $y$ is unique if
$\varphi_{\mathit{prop}}(x)$ holds, and that $y$ does not exist
otherwise. This last statement implies that
\[
\Sig(\Phi',x)=0\qquad\mbox{if }\neg\varphi_{\mathit{prop}}(x).
\]

In the following, let $x\in\{0,1\}^{D}$ be an assignment to the dangling
edges of $\Phi'$ that satisfies $\varphi_{\mathit{prop}}(x)$, and
let $xy\in\{0,1\}^{E(\Phi')\setminus F}$ be its unique extension
to edges not incident with apices, as seen for $\Phi$. We consider
the possible assignments $z\in\{0,1\}^{F}$ to the apex edges such
that $xyz$ satisfies $\Phi'$. Here, while the choice of $y$ was
unique, the choice of $z$ is not unique.

By virtue of $\sigHW{=1}$ at the apex vertices $a_{1}$ and $a_{2}$,
there are unique indices $\tau,\tau'\in A$ such that the edges $a_{1}b_{\tau}$
and $a_{2}b_{\tau'}$ are active in $xy$. By definition of $\act$,
we actually have $\tau=\tau'$, since all elements in $\supp(\act)$
end on $00$ or $11$. We write $\tau^{*}:=\tau=\tau'$ for the unique
``apex-matched'' index, and $b^{*}:=b_{\tau^{*}}$ for the unique
``apex-matched'' vertex. By definition of $\act$, we have 
\[
(xyz)|_{I(b^{*})}\in\{\symNS11,\symNSWE11\}.
\]

It follows that the second component of $\tau^{*}$ must be equal
to $N$, since only vertices in $(\star,N)$ have state $\symNS$
or $\symNSWE$ under $xy$. There are $T$ vertices with signature
$\act$ in row $(\star,N)$, by the balance property of our instance
$\mathcal{T}$ to $\pGrid$, and we can choose any of these vertices
to be apex-matched. To determine the set of such possible choices, we distinguish two cases, depending on whether $(W,N) \in A$ or not.
\begin{description}
	\item [{$(W,N)\notin A:$}] The apex-matched vertex must be in state $\symNS11$
	under $xyz$. It cannot be in state $\symNSWE11$, since only $b_{W,N}$
	can have state $\symNSWE$ among its first four edges, but $b_{W,N}$
	has $\pass$ assigned, since $(W,N)\notin A$. This gives $T$ assignments
	$z$ such that $xyz$ satisfies $\Phi'$. Each of the $T$ assignments
	$xyz$ satisfies $\val_{\Phi'}(xyz)=-1$, because there is (i) one
	vertex in state $\symNSWE00$, which contributes a factor of $-1$
	to $\val_{\Phi'}(xyz)$, and (ii) some number of vertices in states
	$\symEmpty00$, $\symNS00$ and $\symWE00$, which however all contribute
	a unit factor to $\val_{\Phi'}(xyz)$. This implies that $\Sig(\Phi',x)=-T$
	if both $(W,N)\notin A$ and $\varphi_{\mathit{prop}}(x)$ hold.
	\item [{$(W,N)\in A:$}] The apex-matched vertex may be in state $\symNS11$
	or $\symNSWE11$. We make a distinction into these two individual
	sub-cases:
	
	\begin{description}
		\item [{$\symNS11$:}] We proceed as in the case of $(W,N)\notin A$, but
		we have only $T-1$ choices left for the apex-matched vertex, since
		$b_{W,N}$ must have state $\symNSWE$ among its first four edges
		and can thus not be in state $\symNS11$. This gives $T-1$ assignments
		$z$ with $\val_{\Phi'}(xyz)=\pass(\symNSWE)=-1$ for each $z$. (In
		the expression of $\val_{\Phi'}(xyz)$, we ignored the vertices in
		states $\symEmpty00$, $\symNS00$ and $\symWE00$ that contribute
		a unit factor.)
		\item [{$\symNSWE11:$}] Since only $b_{W,N}$ can have state $\symNSWE$
		among its first four edges, the apex-matched vertex must be $b_{W,N}$.
		This gives one assignment $z$, and $\val_{\Phi'}(xyz)=\act(\symNSWE)=1$.
		Again, we ignored unit factors.
	\end{description}
	
	In total, if both $(W,N)\in A$ and $\varphi_{\mathit{prop}}(x)$
	hold, then we obtain 
	\[
	\Sig(\Phi',xyz)=(T-1)\cdot(-1)+1=-T+2
	\]

\end{description}
This proves (\ref{eq: Second matchgate}), and thus Lemma~\ref{lem: cell signatures}.
The proof of Theorem~\ref{thm: apex hard} is completed.

\section{\label{sec: permanent modulo}The permanent modulo $2^{k}$}

We prove Theorem~\ref{thm: perm mod 2^k}, which
asserts $\Wone[\parity]$-hardness of evaluating the permanent mod
$2^{k}$. We reduce from the problem $\pGrid[\parity]$,
the parity version of $\pGrid$ from Definition~\ref{def: GridTiling}.
From a high level, the proof resembles that of Theorem~\ref{thm: apex hard},
but the setting of modular evaluation requires
us to apply linearly
combined signatures in a more intricate way.

\subsection{\label{sub:The-main-idea}The main idea}

Our reduction is based upon the following observation: Let $\A=(n,k,\C,\T)$
be an instance for $\pGrid[\parity]$. For $\omega:\C\to[2]$, recall
the graphs $H_{\omega}$ and the numbers $d(\omega)$ from the last
section. We can rewrite (\ref{eq: Combined GridTiling}) as
\begin{equation}
	2^{|\C|}\cdot\pGrid[\#](\T)=\sum_{\omega:\C\to[2]}T^{d(\omega)}\cdot\perm(H_{\omega}).\label{eq: modulo GridTiling =00003D PerfMatch}
\end{equation}

Theorem~\ref{thm: GridTiling hardness} asserts that computing $\pGrid[\parity](\T)$
is $\Wone[\parity]$-hard. Let $M:=2^{|\C|}$ and assume we could
evaluate $\perm(H_{\omega})$ modulo $2M$ for all $\omega$. Using
arithmetic in $\Z{2M}$, we could then evaluate the entire right-hand-side
of (\ref{eq: modulo GridTiling =00003D PerfMatch}), and this allows
us to compute 
\[
M\cdot\pGrid[\#](\T)\equiv_{2M}\begin{cases}
M & \mbox{if }\pGrid[\#](\T)\mbox{ is odd},\\
0 & \mbox{if }\pGrid[\#](\T)\mbox{ is even.}
\end{cases}
\]
Hence, it seems that we could solve $\pGrid[\parity](\T)$ with an oracle for the
permanent modulo $2M=2^{|\C|+1}$, and we might be tempted to
believe that we just proved Theorem~\ref{thm: perm mod 2^k}.

However, the above argument suffers from a fatal gap: The graphs $H_{\omega}$
from the previous section feature edges of weight $\frac{1}{2}$, a number
that does not exist in the rings $\Z{2^{k}}$ for $k\in\N$. In other words, the
proof fails for the surprisingly philosophical reason that the
instances $H_{\omega}$ constructed in the previous section do not even \emph{exist} modulo $2^{k}$.
More precisely, it is the matchgate $\Gamma_{\act}$ used to realize
the signature $\act$ that features this offending weight, and it
is incurred by the part that we called the \emph{even filter}.
To obtain graphs $H_{\omega}$
that avoid edge-weights with even denominators, we therefore construct cell gates using the signature $\pre$ rather than
its more benign version $\act$. This adds several complications to our
arguments, which we can however handle with a suitable linear combination.

\subsection{\label{sub:Revisiting-the-cell}Revisiting the cell gate}

Let $A\subseteq[n]^{2}$ be fixed in the following, and recall the
gates $\Phi$ and $\Phi'$ from Definition~\ref{def: cell gates}.
Note that $\Phi$ features only occurrences of $\pass$, which is
realized by the matchgate $\Gamma_{\pass}$ on edge-weights $-1$
and $1$. We can therefore also realize this gate modulo $2^{k}$.
This does not apply to the gate $\Phi'(A)$, as the matchgate $\Gamma_{\act}$
realizing $\act$ features the weight $\frac{1}{2}$. We modify $\Phi'(A)$
to a new gate $\Gamma(A)$ by replacing all occurrences of $\act$
with $\pre$.
\begin{defn}
	\label{def: Gamma}For $A\subseteq[n]^{2}$, let the gate $\Gamma(A)$
	on $4n$ dangling edges be defined exactly as the gate $\Phi'(A)$
	from Definition~\ref{def: cell gates}, but replace every occurrence
	of $\act$ by $\pre$.
	
	For all $u,v\in[n]$, let $\alpha_{u,v}$ denote the number of occurrences
	of $\pre$ among vertices $b_{\tau}$ with $\tau\in\{(1,v),\ldots,(u-1,v)\}$.
	Likewise, let $\beta_{u,v}$ denote the number of occurrences of $\pre$
	among vertices $b_{\tau}$ with $\tau\in\{(u+1,v),\ldots,(n,v)\}$.
\end{defn}
Figuratively speaking, $\alpha_{u,v}$ is the number of occurrences
of $\pre$ in the column above $(u,v)$, and $\beta_{u,v}$ is the
number of occurrences below it.
In Section~\ref{sub:Realizing-cell-signatures}, we used the vertical
balance property to ensure that $\alpha_{u,v}+\beta_{u,v}$ is equal
to $T-1$ when $(u,v)\in A$, and equal to $T$ when $(u,v)\notin A$.
In this section, this vertical balance will not be required, but \emph{horizontal}
balance will prove useful instead, for different reasons. For the
remainder of our proofs, we define the following auxiliary polynomials,
for all $u,v,w\in[n]$:
\begin{eqnarray}
	q_{u} & := & \sum_{z\in[n]}\alpha_{u,z}\cdot\beta_{u,z}-{\alpha_{u,z} \choose 2}-{\beta_{u,z} \choose 2},\label{eq: poly q}\\
	p_{u,v,w} & := & (\alpha_{u,v}-\beta_{u,v})\cdot(\beta_{u,w}-\alpha_{u,w}),\label{eq: poly p}\\
	r_{u,v} & := & \sum_{\substack{z\in[n]\setminus\{v\}\\
			(u,z)\in A
		}
	}\beta_{u,z},\label{eq: poly r}\\
	s_{u,v} & := & \sum_{\substack{z\in[n]\setminus\{v\}\\
			(u,z)\in A
		}
	}\alpha_{u,z}.\label{eq: poly s}
\end{eqnarray}
Using these polynomials, we can express the signature of $\Gamma$.
\begin{lem}
	\label{lem: sig Gamma}Let $A\subseteq[n]^{2}$, let $\Gamma=\Gamma(A)$
	and let $x\in\{0,1\}^{4n}$ satisfy $\phione$. Recall the conventions
	from Remark~\ref{rem: apex conventions}, including that we implicitly
	decompose the string $x$ into $x_{N},x_{E},x_{S},x_{W}$.
	\begin{itemize}
		\item If $x_{W}\neq x_{E}$ or $\hw(x_{S})\neq1$, then $\Sig(\Gamma,x)=0$.
		\item If $\phiprop(x)$ is true (i.e., we have $x_{W}=x_{E}$ and additionally $x_{N}=x_{S}$), write
		$u:=x_{W}$ and $v:=x_{N}$, with $u,v\in [n]$. Note that these numbers are well-defined. We call such assignments $x$ \emph{wanted},
		and we have
		\[
		\Sig(\Gamma,x)=\begin{cases}
		q_{u}-r_{u,v}-s_{u,v}-\alpha_{u,v}-\beta_{u,v} & \mbox{if }(u,v)\notin A\\
		q_{u}-r_{u,v}-s_{u,v}+1 & \mbox{if }(u,v)\in A
		\end{cases}
		\]
		
		\item If $\phiprop(x)$ is false (i.e, we have $x_{W}=x_{E}$, but $x_{N}\neq x_{S})$, then
		write $u:=x_{W}$, $v:=x_{N}$, and $w:=x_{S}$. We call such assignments
		$x$ \emph{unwanted,} and we have
		\[
		\Sig(\Gamma,x)=\begin{cases}
		p_{u,v,w} & \mbox{if }(u,v)\notin A,\:(u,w)\notin A\\
		p_{u,v,w}+\alpha_{u,v}-\beta_{u,v} & \mbox{if }(u,v)\notin A,\:(u,w)\in A\\
		p_{u,v,w}+\beta_{u,w}-\alpha_{u,w} & \mbox{if }(u,v)\in A,\:(u,w)\notin A\\
		p_{u,v,w}+\beta_{u,w}-\alpha_{u,w}+\alpha_{u,v}-\beta_{u,v}+1 & \mbox{if }(u,v)\in A,\:(u,w)\in A
		\end{cases}
		\]
		
	\end{itemize}
\end{lem}

The full proof of this lemma requires a somewhat tedious calculation,
which is deferred to Section~\ref{sub: Calculating Gamma }. Note that the entries of $\Sig(\Gamma)$ are polynomials in the indeterminates $\alpha_{u,v}$ and $\beta_{u,v}$ for $u,v\in[n]$

Taking Lemma~\ref{lem: sig Gamma} for granted at the moment, we
note that the gate $\Gamma$ essentially discriminates between six
different assignment types, depending on whether $x$ is wanted (giving
$2$ types) or unwanted (giving $4$ types, depending on whether $(x_{W},x_{N})$
and $(x_{W},x_{S})$ are each contained in $A$). However, the actual
value of $\Sig(\Gamma,x)$ is \emph{not constant} for each of the
six types, as it depends on $u,v,w$ and the concrete values for $\alpha_{u',v'}$
and $\beta_{u',v'}$ for all $u',v'\in[n]$. Compare this to the gate
$\Phi'$ from Section~\ref{sub: signature phi'}, which attains one
of the three fixed values $\{0,-T,-T+2\}$ due to vertical balance. It turns out that the simple balance argument used in the last section does not work in this setting; our technical efforts in the remainder of the proof therefore aim at the following two goals:

\begin{description}
	\item [{Goal~1:}] \quad Ensure that the four unwanted cases (as defined above) cancel out.
	\item [{Goal~2:}] \quad Ensure that the two wanted cases (as defined above) do not depend upon the actual
	value of $(x_{W},x_{N})$, but only on the information whether $(x_{W},x_{N})\in A$
	or $(x_{W},x_{N})\notin A$.
\end{description}

In the following, we show how to attain these goals by considering a particular linear combination of matchgate signatures that could be considered as the ``derivative'' of a matchgate.

\subsection{Linear combinations via discrete derivatives}

Recall the construction of $\Gamma$ from Definition~\ref{def: Gamma}. In the following, we construct
a gate $\Gamma_{\uparrow}$ from $\Gamma$ by adding several ``dummy'' vertices. Then we consider the difference
\[\Sig(\Gamma_{\uparrow})-\Sig(\Gamma).\] The gate $\Gamma_{\uparrow}$
is obtained from $\Gamma$ by adding dummy rows of vertices with signature
$\pre$, and this allows us to obtain $\Sig(\Gamma_{\uparrow})$ by a
simple substitution on the indeterminates of $\Sig(\Gamma)$.
\begin{defn}
	\label{def: Gamma up down}We define a \emph{dummy gate} as in Figure~\ref{fig: dummy-gamma}:
	Starting from a vertex with signature $\pre$, add several vertices of signature
	$\sigHW{=1}$ to its western and eastern dangling edges to force these
	edges to be inactive, as shown in the left part of the figure. We then define a \emph{dummy row} by arranging
	$n$ dummy gates horizontally as shown in the right part of the figure.
	\begin{figure}
		\begin{centering}
			\includegraphics[width=0.75\textwidth]{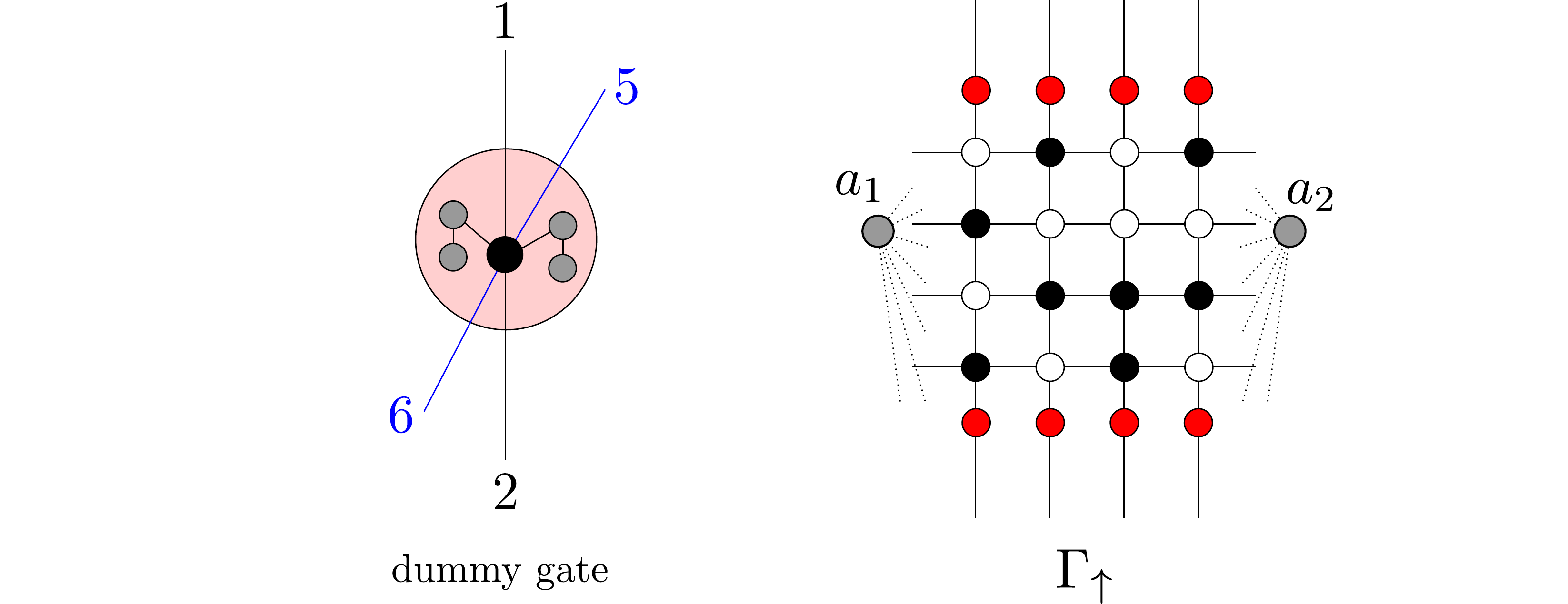}
			\par\end{centering}
		
		\protect\caption{\label{fig: dummy-gamma}A dummy gate is shown on the left. On the
			right, we see $\Gamma_{\uparrow}$, which is obtained from $\Gamma$
			by adding rows of dummy gates, shown red. Each gray vertex is assigned
			$\protect\sigHW{=1}$, and the apices connect to all black vertices
			(assigned $\protect\pre$) and all red vertices (whose signature is
			realized by the dummy gate). White vertices are assigned $\protect\pass$,
			and they are not adjacent to apices.}
	\end{figure}
	
	Starting from $\Gamma$, define a gate $\Gamma_{\uparrow}$ by
	adding a dummy row above the row $(1,\star)$, and a dummy row below
	the row $(n,\star)$, as shown in Figure~\ref{fig: dummy-gamma}.
	We connect apex $a_{1}$ to the dangling edge $5$ of each dummy gate,
	and $a_{2}$ to the dangling edge $6$.
\end{defn}
Furthermore, we define algebraic manipulations on multivariate polynomials
that capture the effect of introducing dummy rows into $\Gamma$ as described above.
\begin{defn}
	\label{def: algebra on p}Let $p$ be any multivariate polynomial
	on the indeterminates $\alpha_{u,v}$ and $\beta_{u,v}$ for $u,v\in[n]$.
	Write $x\gets y$ for the operation of substituting $x$ with $y$
	in $p$. Then we define $p_{\uparrow}$ to be the polynomial obtained from $p$ after performing the substitutions $\alpha_{u,v}\gets\alpha_{u,v}+1$ and $\beta_{u,v}\gets\beta_{u,v}+1$ for all $u,v\in[n]$.
	
	We also define the following discrete derivative operator $D$ on such polynomials $p$:
	\[
	D(p):=p_{\uparrow}-p.
	\]
\end{defn}
The following is then easily observed:
\begin{lem}
	We have \[\Sig(\Gamma_{\uparrow})=(\Sig(\Gamma))_{\uparrow},\] and
	in particular, we have
	\[D(\Sig(\Gamma))=\Sig(\Gamma_{\uparrow})-\Sig(\Gamma).\]
\end{lem}
Note that the operator $D$ indeed resembles a derivative: We have linearity
by $D(p+q)=D(p)+D(q)$, and applying $D$ to a polynomial $p$ of
degree $d$ gives one of degree $d-1$. We will use these properties of $D$
to effect two useful modifications on the polynomials in (\ref{eq: poly q})-(\ref{eq: poly r}),
and thus ultimately on $\Sig(\Gamma)$. These correspond to the two
goals described at the end of Section~\ref{sub:Revisiting-the-cell}.
\begin{enumerate}
	\item Concerning the first goal, our choice of $D$ ensures that ``unwanted''
	polynomials vanish under $D$. For instance, for all $u,v,w\in[n]$,
	the polynomial $p_{u,v,w}$ from (\ref{eq: poly p}) maps to
	\begin{eqnarray}
		D(p_{u,v,w}) & = & ((\alpha_{u,v}+1)-(\beta_{u,v}+1))\cdot((\beta_{u,w}+1)-(\alpha_{u,w}+1))\nonumber \\
		&  & -(\alpha_{u,v}-\beta_{u,v})\cdot(\beta_{u,w}-\alpha_{u,w})\nonumber \\
		& = & 0.\label{eq: Dp is zero}
	\end{eqnarray}
	By our calculation of $\Sig(\Gamma)$, this implies that $D(\Sig(\Gamma))$
	vanishes on assignments $x$ with $x_{N}\neq x_{S}$ and $(x_{W},x_{N})\notin A$
	and $(x_{W},x_{S})\notin A$. The other unwanted cases will be handled
	by similar arguments.
	\item Under the operator $D$, linear terms, such as $\alpha_{u,v}$ or
	$\beta_{u,v}$ for $u,v\in[n]$, are mapped to 
	\begin{eqnarray}
		D(\alpha_{u,v}) & = & (\alpha_{u,v}+1)-\alpha_{u,v}\ =\ 1,\label{eq: Da}\\
		D(\beta_{u,v}) & = & (\beta_{u,v}+1)-\beta_{u,v}\ =\ 1.\label{eq: Db}
	\end{eqnarray}
	This helps us to attain the second goal, since the original terms
	depend on the concrete values of $\alpha_{u,v}$ or $\beta_{u,v}$
	in $A$, whereas the resulting constants do not. It will also turn
	out that only linear terms need to be considered.
\end{enumerate}
In the following, we show that $D(\Sig(\Gamma))$ essentially realizes
the function $g_{\kappa}$, up to some additive term on assignments
$x$ with $\varphi_{\mathit{prop}}$. This allows us to write $g_{\kappa}$ as a linear combination of the
matchgate signatures $\Sig(\Gamma_{\uparrow})$ and $\Sig(\Gamma)$.
As a technical requirement, we use Lemma~\ref{lem: GridTiling balance}
to ensure that the set $A$ in the definition of $\Gamma=\Gamma(A)$
is horizontally balanced.

\begin{lem}
	\label{lem: sig D}
	Assume
	the existence of a number $T\in\N$ such that $A$ features exactly
	$T$ elements of type $(u,\star)$, for all $u\in[n]$. Let $\Gamma=\Gamma(A)$
	and write $D:=D(\Sig(\Gamma))=\Sig(\Gamma_{\uparrow})-\Sig(\Gamma)$.
	We then have
	\[
	D=\begin{cases}
	0 & \quad\mbox{if }\neg\varphi_{\mathit{prop}}(x)\\
	\begin{cases}
	n-2T-2 & (x_{W},x_{N})\notin A\\
	n-2T+2 & (x_{W},x_{N})\in A
	\end{cases} & \quad\mbox{if }\varphi_{\mathit{prop}}(x)
	\end{cases}
	\]
\end{lem}
\begin{proof}
	We prove the identity using linearity of $D$. For all $u,v,w\in[n]$,
	consider the effect of $D$ on the polynomials from (\ref{eq: poly q})-(\ref{eq: poly s}).
	For instance, we have seen in (\ref{eq: Dp is zero}) and (\ref{eq: Da})-(\ref{eq: Db})
	that 
	\begin{eqnarray*}
		D(p_{u,v,w}) & = & 0,\\
		D(\alpha_{u,v})=D(\beta_{u,v}) & = & 1.
	\end{eqnarray*}
	Likewise, we can show that 
	\begin{eqnarray*}
		D(q_{u}) & = & \sum_{v\in[n]}1\ =\ n,\\
		D(r_{u,v})=D(s_{u,v}) & = & \sum_{\substack{z\in[n]\setminus\{v\}\\
				(u,z)\in A
			}
		}1\ =\ \begin{cases}
		T & (u,v)\notin A,\\
		T-1 & (u,v)\in A.
	\end{cases}
\end{eqnarray*}
Together with linearity of $D$ and the expression of $\Sig(\Gamma)$
from \ref{lem: sig Gamma}, this proves the claim by a simple calculation
for each of the six assignment types.\end{proof}
\begin{cor}
	\label{cor: combine modulo}Write $S:=n-2T-2$ and recall the matchgate
	$\Phi$ from Section~\ref{sub: signature phi} with
	\[
	\Sig(\Phi,x)=\begin{cases}
	1 & \mbox{if }\varphi_{\mathit{prop}}(x),\\
	0 & \mbox{otherwise}.
	\end{cases}
	\]
	Then the following linear combination realizes the signature $g_{\kappa}$:
	\[
	\frac{D-S\cdot\Sig(\Phi)}{4}=\frac{\Sig(\Gamma_{\uparrow})-\Sig(\Gamma)-S\cdot\Sig(\Phi)}{4}.
	\]
	Note that each of the constituent gates $\Gamma_{\uparrow}$, $\Gamma$
	and $\Phi$ has at most two apices and features only edge-weights
	from the set $\{-1,1\}$. Furthermore, each of these gates admits a $2$-coloring.
\end{cor}

Using Corollary~\ref{cor: combine modulo}, we can complete the proof
of Theorem~\ref{thm: perm mod 2^k}. Recall that we aim at a reduction from $\pGrid[\parity]$ to the permanent modulo $2^k$.
\begin{proof}
	[Proof of Theorem \ref{thm: perm mod 2^k}]Let $\A=(n,k,\C,\T)$ be
	an instance for the $\Wone[\oplus]$-complete problem $\pGrid[\parity]$. To prove the lower bound under $\ETH[\parity]$,
	we may assume $|\C|=\O(k)$ by Theorem~\ref{thm: GridTiling hardness}.
	Furthermore, by horizontal balance via Lemma~\ref{lem: GridTiling balance}, we may assume
	that we are given a number $T\in\N$ such that $|\T(\kappa)\cap(u,\star)|=T$
	for all $\kappa\in\C$ and $u\in[n]$.
	
	Recall Definition~\ref{def: f and g} and Lemma~\ref{lem: Holant =00003D GridTilings}
	of Section~\ref{sub:Global-construction}: These allow us to compute
	a signature graph $G$ with signatures $f_{\kappa}$ at $\kappa\in[k]^{2}\setminus\C$
	and signatures $g_{\kappa}$ at $\kappa\in\C$ such that
	\[
	\pGrid[\#](\A)=\Holant(G).
	\]
	
	As shown in Lemma~\ref{lem: cell signatures}, we can realize $f_{\kappa}$
	by the planar matchgate $\Phi$ on edge-weights $\{-1,1\}$. Furthermore,
	as shown in Lemma \ref{lem: sig D}, we can realize
	$g_{\kappa}$ for each $\kappa\in\C$ as the linear combination of
	three $2$-apex matchgates on edge-weights $\{-1,1\}$: Let $\Gamma_{\kappa}:=\Gamma(\T(\kappa))$
	be as in Definition~\ref{def: Gamma}, and let $\Gamma_{\kappa,\uparrow}$
	be obtained from $\Gamma_{\kappa}$ as in Definition~\ref{def: Gamma up down}.
	Then, similarly to the proof of Theorem~\ref{thm: apex hard}, we
	obtain with Lemma \ref{lem: sig D} and Lemma~\ref{lem: combined signature lemma}
	about the linear combinations of signatures that
	\begin{equation}
		4^{|\C|}\cdot\Holant(G)=\sum_{\omega:\C\to[3]}(-1)^{d(\omega)}\cdot(-S)^{e(\omega)}\cdot\PerfMatch(H_{\omega}).\label{eq: Modulo Combination}
	\end{equation}
	Here, for each $\omega:\C\to[3]$, the number $d(\omega)$ is defined
	to be the number of $2$-entries in $\omega$, and $e(\omega)$ is
	the number of $3$-entries. The graph $H_{\omega}$ is obtained as
	follows: For $\kappa\in[k]^{2}\setminus\C$, insert the matchgate
	$\Phi$ at the cell vertex $c_{\kappa}$. For all $\kappa\in\C$,
	insert $\Gamma_{\kappa,\uparrow}$ or $\Gamma_{\kappa}$ or $\Phi$
	at $c_{\kappa}$ if $\omega(\kappa)$ is $1$ or $2$ or $3$, respectively.
	
	Let $M:=2^{2|\C|}$. With an oracle for computing $\PerfMatch(H_{\omega})$
	modulo $2M$ for all $\omega$, we can compute the right-hand side
	of (\ref{eq: Modulo Combination}) modulo $2M$ via arithmetic in $\Z{2M}$.
	We then obtain the value (modulo $2M$) of 
	\[
	M\cdot\Holant(G)=M\cdot\pGrid[\#](\A)\equiv_{2M}\begin{cases}
	M & \mbox{if }\pGrid[\#](\A)\mbox{ odd},\\
	0 & \mbox{if }\pGrid[\#](\A)\mbox{ even.}
	\end{cases}
	\]

	Each graph $H_{\omega}$ is bipartite, has at most $2|\C|=\O(k)$
	apices, and the computation is modulo $2M=2^{\O(k)}$. We have thus
	shown a parameterized Turing reduction from $\pGrid[\parity]$ to
	the evaluation of the permanent on $\O(k)$-apex graphs modulo $2^{\O(k)}$.
	Since Theorem~\ref{thm: GridTiling hardness} asserts the $\Wone[\oplus]$-completeness of the former problem, the theorem
	follows.
\end{proof}

\subsection{\label{sub: Calculating Gamma }Calculating the signature of $\Gamma$}

In the remainder of this section, we prove Lemma~\ref{lem: sig Gamma}.
Let $x\in\{0,1\}^{4n}$ be an assignment to the dangling edges of
$\Gamma$. The statement of the lemma is shown by inspecting the possible
satisfying extensions of $x$, as we did when calculating $\Sig(\Phi')$.
To understand the following proof, we therefore recommend recalling
Section~\ref{sub: signature phi'}, since that section contains a similar, yet substantially simpler argument. 

Let $F\subseteq E(\Gamma)$ denote the edges of $\Gamma$ that are
incident with apices. Given $x$, let $xyz\in\{0,1\}^{E(\Gamma)}$
be an assignment extending $x$ such that $\Sig(\Gamma,xyz)\neq0$,
with $y\in\{0,1\}^{E(\Gamma)\setminus F}$ and $z\in\{0,1\}^{F}$.
Due to $\sigHW{=1}$ at the apex vertices $a_{1}$ and $a_{2}$ of
$\Gamma$, there are apex-matched indices $\tau_{1},\tau_{2}\in A$
and apex-matched vertices $b_{1}:=b_{\tau_{1}}$ and $b_{2}:=b_{\tau_{2}}$
such that $a_{1}b_{1}$ and $a_{2}b_{2}$ are active in $xyz$. However,
opposing Section~\ref{sub: signature phi'}, it may well be that
$\tau_{1}\neq\tau_{2}$, and this makes our calculations somewhat
more difficult. In particular, the assignment $y$ is no longer uniquely determined
by $x$.

For each assignment $x$, we partition the satisfying extending assignments
$xyz$ to $\Gamma$ into six partition classes $\{\mathcal{P}_{i}(x)\}_{i\in[6]}$,
corresponding to the states of the (at most two distinct) apex-matched
vertices. More precisely, for $i\in[6]$, we let
\[
\mathcal{P}_{i}(x):=\{xyz\in\{0,1\}^{E(\Gamma)}\mid xyz|_{I(b_{1})}\mbox{ and }xyz|_{I(b_{2})}\mbox{ are as in row }i\mbox{ of Table \ref{tab: calculate Gamma}}\}.
\]
\begin{table}
	\begin{centering}
		\begin{tabular}{c|c|c||c|c|c|c|}
			& $(u,v)\notin A$  & $(u,v)\in A$ & $\substack{(u,v)\notin A\\
				(u,w)\,\notin A
			}
			$ & $\substack{(u,v)\notin A\\
				(u,w)\in A
			}
			$ & $\substack{(u,v)\in A\\
				(u,w)\notin A
			}
			$ & $\substack{(u,v)\in A\\
				(u,w)\in A
			}
			$\tabularnewline
			\hline 
			$\symNSWE11$ & $0$ & $1$ & $0$ & $0$ & $0$ & $0$\tabularnewline
			\hline 
			$\symNS11$ & $-\alpha_{u,v}-\beta_{u,v}$ & $-\alpha_{u,v}-\beta_{u,v}$ & $0$ & $0$ & $0$ & $0$\tabularnewline
			\hline 
			\hline 
			$\symN10,\symS01$ & $q_{u}$ & $q_{u}$ & $p_{u,v,w}$ & $p_{u,v,w}$ & $p_{u,v,w}$ & $p_{u,v,w}$\tabularnewline
			\hline 
			$\symN10,\symSWE01$ & $-r_{u,v}$ & $-r_{u,v}+\alpha_{u,v}$ & $0$ & $\alpha_{u,v}-\beta_{u,v}$ & $0$ & $\alpha_{u,v}-\beta_{u,v}$\tabularnewline
			\hline 
			$\symNWE10,\symS01$ & $-s_{u,v}$ & $-s_{u,v}+\beta_{u,v}$ & $0$ & $0$ & $\beta_{u,w}-\alpha_{u,w}$ & $\beta_{u,w}-\alpha_{u,w}$\tabularnewline
			\hline 
			$\symNWE10,\symSWE01$ & $0$ & $0$ & $0$ & $0$ & $0$ & $1$\tabularnewline
			\hline 
		\end{tabular}
		\par\end{centering}
	
	\caption{\label{tab: calculate Gamma}The six assignment types of the cell
		are listed as columns, and the possible states of the (at most two)
		apex-matched vertices are listed as rows. The signature of $\Gamma$
		on each of the six assignment types is given as the sum of the elements
		in the corresponding column. Note that the table is divided into four
		quadrants. We have essentially already calculated the top left quadrant
		in Section~\ref{sub:Computing-the-signatures} when we calculated
		$\protect\Sig(\Phi')$.}
\end{table}
Note that $b_{1}$ and $b_{2}$ depend upon the assignment $xyz$.
To give an example, in row $1$, and thus in class $\mathcal{P}_1$, we consider extending assignments $xyz$ that have only one vertex with active edges leading to an apex, and the local assignment at this vertex reads $\symNSWE11$. More formally, we have 
\[
b_{1}=b_{2}\ \wedge\ xyz|_{I(b_{1})}=\symNSWE11.
\]
As another example, in row $3$, we have 
\[
b_{1}\neq b_{2}\ \wedge\ xyz|_{I(b_{1})}=\symN10\ \wedge\ xyz|_{I(b_{2})}=\symS01.
\]
It is evident that, given $x\in\{0,1\}^{4n}$, we have 
\begin{equation}
	\Sig(\Gamma,x)=\sum_{i\in[6]}\underbrace{\sum_{xyz\in\mathcal{P}_{i}(x)}\val_{\Gamma}(xyz)}_{=:P_{i}(x)}.\label{eq: sum states}
\end{equation}

In Table~\ref{tab: calculate Gamma}, we calculate $P_{i}(x)$ for
all $i\in[6]$ and all six types of assignments $x$ to dangling edges distinguished
by the signature: The entry in this table at row $i\in[6]$ and column
$j\in[6]$ denotes the number $P_{i}(x)$ on assignments $x$ of the $j$-th
type. Note that the table is divided into four quadrants, as indicated
by the double lines in Table~\ref{tab: calculate Gamma}. In Section~\ref{sub: signature phi'},
we have essentially already calculated the values in the top left
quadrant. In the following, we calculate the remaining quadrants.

Before doing so, we first need to make some general observations:
In each satisfying assignment $xyz$ extending $x$, all western and
eastern edges of vertices in the row $(x_{W},\star)$ are active,
and no other western and eastern edges are active. This is because
for any vertex $b\in V(\Gamma)\setminus\{a_{1},a_{2}\}$, the signatures
$\pass$ and $\pre$ imply that the assignment $xy|_{I(b)}$ has one
of the states 
\begin{equation}
	\underbrace{\symEmpty,\symNS,\symWE,\symNSWE}_{\mathrm{``tame"}},\underbrace{\symN,\symS,\symNWE,\symSWE}_{\mathrm{``wild"}}.\label{eq: states}
\end{equation}
In each such state, be it tame or wild, the western incident edge
is active iff the eastern edge is active as well. By an argument as
in Section~\ref{sub: signature phi}, this implies $x_{W}=x_{E}$
for the assignment $x$. Note that a similar statement from north
to south is not necessarily true, as witnessed by vertices in a ``wild''
state. 

If $b_{1}\neq b_{2}$, this implies $xyz|_{I(b_{1})}\in\{\symN10,\symNWE10\}$
and $xyz|_{I(b_{2})}=\{\symS01,\symSWE01\}$. Because all other vertices
are in tame states and thus enforce equality on their northern and
southern dangling edges, the vertex $b_{1}$ ``shoots'' a ray of
active vertical edges to the north (transmitted by vertices in state
$\symNS$, $\symNSWE$, $\symNS00$, $\symNSWE00$). This ray may
either leave the cell, or it hits $b_{2}$. We conclude that, for
any column $j\in[n]$, 
\begin{itemize}
	\item $x_{N}(j)=x_{S}(j)$ iff column $(\star,j)$ contains neither $b_{1}$
	nor $b_{2}$, or it contains both,
	\item $x_{N}(j)=1\ \wedge\ x_{S}(j)=0$ iff column $(\star,j)$ contains
	$b_{1}$ but not $b_{2}$, 
	\item $x_{N}(j)=0\ \wedge\ x_{S}(j)=1$ iff column $(\star,j)$ contains
	$b_{2}$ but not $b_{1}$.
\end{itemize}
We are now ready to calculate the remaining quadrants of Table~\ref{tab: calculate Gamma}.
Recall that we use the abbreviations $u:=x_{W}$, $v:=x_{N}$ and $w:=x_{S}$.

\paragraph*{Top right quadrant: }

If $x_{N}\neq x_{S}$, then $\mathcal{P}_{1}(x)=\mathcal{P}_{2}(x)=\emptyset$.
This is because all vertices in assignments $xyz\in\mathcal{P}_{1}(x)\cup\mathcal{P}_{2}(x)$
are in tame states, which would imply $x_{N}=x_{S}$. This explains
all zeros in the top right quadrant of Table~\ref{tab: calculate Gamma}.

\paragraph*{Bottom right quadrant (0/1 entries): }

If $x_{N}\neq x_{S}$ and $(x_{W},x_{N})\notin A$, then no satisfying
assignment has a vertex in state $\symNWE$: By our general observation,
the index of this vertex would be $(x_{W},x_{N})$, but this vertex
has no adjacent apex, since $(x_{W},x_{N})\notin A$, and it can thus
only be in a tame state. Likewise, if $(x_{W},x_{S})\notin A$, then
no satisfying assignment has a vertex in state $\symSWE$. This explains
all zeros in the bottom right quadrant of Table~\ref{tab: calculate Gamma},
and it also explains the bottom right entry of $1$.

\paragraph*{Bottom right quadrant (other entries): }

By our general observation, the vertex $b_{1}$ must be located in
the column $(\star,v)$ and $b_{2}$ must be located in the column
$(\star,w)$ . 

Consider the third row in the right quadrant and Figure~\ref{fig: states-btm-right-quadrant}. Because of the states
of $b_{1}$ and $b_{2}$, neither of them is on the horizontal path
$u$. This gives $\alpha_{u,v}+\beta_{u,v}$ choices for $b_{1}$.
When $b_{1}$ is above $(u,v)$, there are $\alpha_{u,v}$ possibilities,
and the northbound ray emitted by $b_{1}$ does not cross the horizontal
path in $(u,\star)$ described in the general observations. When $b_{1}$
is below $(u,v)$, there are $\beta_{u,v}$ possibilities, and the
northbound ray crosses the horizontal path in $(u,\star)$, so the
vertex at $(u,v)$ contributes a factor $-1$ from $\pass(\symNSWE)$
or $\pre(\symNSWE00)$.
By a similar analysis for $b_{2}$
as for $b_{1}$, we obtain four cases, shown in Figure~\ref{fig: states-btm-right-quadrant} and we see that, for inputs $x$ of the third type in Table~\ref{tab: calculate Gamma}, we have 
\begin{eqnarray*}
	P_{3}(x) & = & \alpha_{u,v}\cdot\beta_{u,w}-\alpha_{u,v}\cdot\alpha_{u,w}-\beta_{u,v}\cdot\beta_{u,w}+\beta_{u,v}\cdot\alpha_{u,w}\\
	& = & (\alpha_{u,v}-\beta_{u,v})\cdot(\beta_{u,w}-\alpha_{u,w})\\
	& = & p_{u,v,w}
\end{eqnarray*}

The calculation of the remaining rows of Table~\ref{tab: calculate Gamma}
is similar, except that $b_{1}$ or $b_{2}$ may appear on the horizontal
path $(u,\star)$ by the $\symNWE$ or $\symSWE$ state, so only one
or fewer factors of $p_{u,v,w}$ remain.

\begin{figure}
	\begin{centering}
		\includegraphics[width=0.65\textwidth]{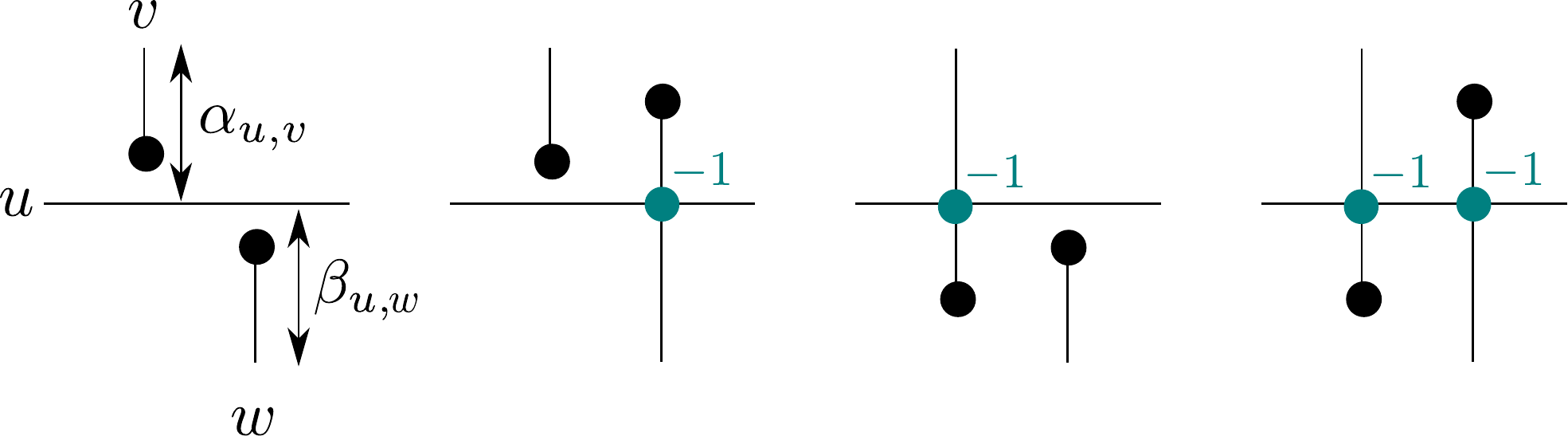}
		\par\end{centering}
	
	\caption{\label{fig: states-btm-right-quadrant} Relevant states in the bottom right quadrant. The vertices $b_1$ and $b_2$ are shown as black dots, crossings with the horizontal path are shown as turquoise dots.}
\end{figure}

\paragraph*{Bottom~left~quadrant~(zero~entries): }

The argument for the zero entries in the bottom right quadrant applies here as well.

\paragraph*{Bottom~left~quadrant~(other~entries): }

It can be verified that $b_{1}$ and $b_{2}$
must be located in the same column, as otherwise it would be impossible to have $x_N = x_S$. In particular, either they are in some column
$(\star,j)$ with $j\neq v$, or they are in the column $(\star,v)$.
We calculate the weighted sum over the relevant extensions in Table~\ref{tab: bottom left}, and then use it to get
the bottom left quadrant of Table~\ref{tab: calculate Gamma}. To verify the completeness of our reasoning, we advise to tick the corresponding cells of the table while reading.

Let us assume first that $b_{1}$ and $b_{2}$ appear in a column $(\star,j)$ of $\Gamma$ with $j\neq v$. These situations are covered in columns 1, 2, 4, and 5 of Table~\ref{tab: bottom left}. Then, after fixing the positions of $b_{1}$ and $b_{2}$, the unique possible assignment realizing this choice contains the horizontal path $(u,\star)$, a vertical path $(\star,v)$ and a path connecting $b_{1}$ and $b_{2}$. The vertex at $(u,v)$
yields the value $-1$, since it is in state $\symNSWE$ or $\symNSWE 00$. Whether the vertex at $(u,j)$ also yields $-1$ depends on
whether the line segment $b_{1}b_{2}$ crosses the horizontal path
$(u,\star)$. 

Consider the first row of Table~\ref{tab: bottom left} for columns with $j \neq v$.
When $b_{1}$ and $b_{2}$ are in states $\symN10$ and $\symS01$
respectively, there are $\alpha_{u,j}\beta_{u,j}$ choices for $b_{1}$ and $b_{2}$ such that the line segment $b_{1}b_{2}$ crosses the horizontal path (and in this case, we have two crossings, each of which yields a factor $-1$). There are ${\alpha_{u,j} \choose 2}+{\beta_{u,j} \choose 2}$
choices of $b_{1}$ and $b_{2}$ such that the crossing
does not occur (yielding one crossing in total and a factor $-1$). Hence, the total sum over extensions to $x$ with $b_{1}$ and $b_{2}$ in states $\symN10$ and $\symS01$ is equal to
\[
t_{j}=\alpha_{u,j}\beta_{u,j}-{\alpha_{u,j} \choose 2}-{\beta_{u,j} \choose 2}.
\]

We observe that no extension to $x$ can have the vertices $b_1$ and $b_2$ in states $\symNWE10$ and $\symSWE01$, as these states would force the vertices $b_1$ and $b_2$ to appear in different columns of $\Gamma$. Hence, the number of extensions in row 4 are all zero.
Note also that, in columns 1, 3, and 4, no states other than $\symN10$ and $\symS01$  can appear: Every other state would require $(u,j) \in A$, since only such vertices can possibly be in wild states.

The calculations so far have settled columns 1 and 4; we now consider column 2. 
If and only if $b_{2}$ is located on $(u,j)$, then the vertices $b_1$ and $b_2$ are in states $\symN10,\symSWE01$.
Then the vertex $b_2$ at $(u,j)$ gives $\pre(\symSWE01)=1$, and $b_1$ gives $\pre(\symN10)=1$. 
The vertex at $(u,v)$ is in state $\symNSWE$ or $\symNSWE 00$ and consequently yields $-1$.
We observe that there are $\beta_{u,j}$ choices for $b_1$. This settles row 2 of column 2.
A symmetric argument applies in row 3 of column 2, when the vertices $b_1$ and $b_2$ are in states $\symNWE10,\symS01$.

The same argument applies to column 5, 
since both $b_1$ and $b_2$ do not appear in the $v$-th column of $\Gamma$.
This settles all columns with $j\neq v$; we will henceforth consider the case $j=v$ as in columns 3 and 6. 
In these columns, the vertices $b_1$ and $b_2$ must be situated in column $(\star,v)$ of $\Gamma$. Furthermore, we again have the horizontal path passing through row $(u,\star)$.

Consider row 1, 
corresponding to states $\symN10,\symS01$. Here, it is irrelevant whether 
$(u,v)\in A$ or not, since none of $b_1$ or $b_2$ can be located at $(u,v)$, as the horizontal path could otherwise not pass through these vertices. There are $\alpha_{u,j}\beta_{u,j}$ possible positions
for $b_{1}$ and $b_{2}$ such that $b_1$ lies above
the horizontal path $(u,\star)$ and $b_2$ lies below it. In both situations, no crossing occurs. Furthermore, there are ${\alpha_{u,j} \choose 2}+{\beta_{u,j} \choose 2}$
possible positions for $b_{1}$ and $b_{2}$ such that both lie above or both lie below the horizontal path, introducing precisely one crossing with the path. Hence, the weighted sum over extensions is again given by $t_{j}$, with $j = v$. 

This settles column 3; it remains to consider column 6.
Consider its second row. Because $b_{2}$
is in state $\symSWE01$, it is located at $(u,v)$, and shoots a ray
to the south. There are $\alpha_{u,v}$ positions left for $b_{1}$
to shoot a ray to the north. Similarly, the third entry is $\beta_{u,v}$. It is important to note here that no crossing occurs, as opposed to, say, column 5.

We have now calculated all entries of the table. If we sum the first $3$ columns and the last $3$ columns, respectively, we get the bottom
left quadrant of Table~\ref{tab: calculate Gamma}. (Note that each block of $3$ columns actually corresponds to $n$ choices for $j$, so each sum involves $n$ terms.)

\begin{table}
	\begin{centering}
		\begin{tabular}{c|c|c|c|c|c|c|}
			
			states of $b_1$, $b_2$
			
			& $\substack{(u,v)\notin A\\
				j\neq v\\
				(u,j)\notin A
			}
			$  & $\substack{(u,v)\notin A\\
				j\neq v\\
				(u,j)\in A
			}
			$  & $\substack{(u,v)\notin A\\
				j = v
			}
			$ & $\substack{(u,v)\in A\\
				j\neq v\\
				(u,j)\notin A
			}
			$  & $\substack{(u,v)\in A\\
				j\neq v\\
				(u,j)\in A
			}
			$  & $\substack{(u,v)\in A\\
				j = v
			}
			$\tabularnewline
			\hline 
			$\symN10,\symS01$ & $t_{j}$ & $t_{j}$ & $t_{j}$ & $t_{j}$ & $t_{j}$ & $t_{j}$\tabularnewline
			\hline 
			$\symN10,\symSWE01$ & $0$ & $-\beta_{u,j}$ & $0$ & $0$ & $-\beta_{u,j}$ & $\alpha_{u,v}$\tabularnewline
			\hline 
			$\symNWE10,\symS01$ & $0$ & $-\alpha_{u,j}$ & $0$ & $0$ & $-\alpha_{u,j}$ & $\beta_{u,v}$\tabularnewline
			\hline 
			$\symNWE10,\symSWE01$ & $0$ & $0$ & $0$ & $0$ & $0$ & $0$\tabularnewline
			\hline 
		\end{tabular}
		\par\end{centering}
	
	\caption{\label{tab: bottom left}A detailed table of the bottom left quadrant
		of Table~\ref{tab: calculate Gamma}.}
\end{table}

\paragraph*{Conclusion of the calculation.}

This explains all entries of Table~\ref{tab: calculate Gamma}. Given
an assignment $x$ having one of the types indicated in the columns
of Table~\ref{tab: calculate Gamma}, the value $\Sig(\Gamma,x)$
is then obtained by summing along the corresponding column as in (\ref{eq: sum states}).

\section*{Acknowledgement}
We wish to thank an anonymous reviewer for extensive and helpful comments on the submitted version.

\bibliographystyle{plain}
\bibliography{References}

\end{document}